\documentclass[a4paper,UKenglish,cleveref, autoref, thm-restate]{lipics-v2021}

\pdfoutput=1 
\hideLIPIcs  

\title{Information Propagation and Contraction in Functional Interpretations}

\author{Chuangjie {Xu}}{}{}{}{}

\authorrunning{C. Xu}

\Copyright{Chuangjie Xu}



\acknowledgements{We are grateful to Tom de Jong, Fredrik Nordvall Forsberg, Paulo Oliva, and Helmut Schwichtenberg for valuable discussions on this work and for their careful proofreading of the manuscript.}

\nolinenumbers 

\EventEditors{John Q. Open and Joan R. Access}
\EventNoEds{2}
\EventLongTitle{42nd Conference on Very Important Topics (CVIT 2016)}
\EventShortTitle{CVIT 2016}
\EventAcronym{CVIT}
\EventYear{2016}
\EventDate{December 24--27, 2016}
\EventLocation{Little Whinging, United Kingdom}
\EventLogo{}
\SeriesVolume{42}
\ArticleNo{23}

\usepackage{booktabs}
\usepackage{mathpartir}
\usepackage{bbm}

\newcommand{\T}{\mathsf{T}}
\newcommand{\HAo}{\mathsf{HA}^\omega}
\newcommand{\NHAo}{\mathsf{N}\text{-}\HAo}
\newcommand{\EHAo}{\mathsf{E}\text{-}\HAo}
\newcommand{\WEHAo}{\mathsf{WE}\text{-}\HAo}
\newcommand{\aff}{\mathsf{aff}}
\newcommand{\NHAoa}{\NHAo_\aff}

\newcommand{\OHAo}{\HAo_{\mathrm{obs}}}

\newcommand{\pair}[2]{\langle #1, #2 \rangle}
\newcommand{\pr}{\mathsf{pr}}

\newcommand{\N}{\mathbb{N}}
\newcommand{\suc}{\mathsf{S}}
\newcommand{\rec}{\mathsf{rec}}
\newcommand{\Two}{\mathbbm{2}}
\newcommand{\sft}{\mathsf{t}}
\newcommand{\sff}{\mathsf{f}}
\newcommand{\sfif}{\mathsf{if}}
\newcommand{\One}{\mathbbm{1}}
\newcommand{\sys}{\mathcal{S}}
\newcommand{\src}{\mathcal{A}}

\newcommand{\eqdef}{\mathrel{:\equiv}}
\newcommand{\eqsubst}{\mathrel{:=}}

\newcommand{\nucleus}{\sharp}
\newcommand{\compat}{\triangleleft}
\newcommand{\nule}{\mathbin{\boldsymbol{\preccurlyeq}}}
\newcommand{\ctr}{\mathsf{ctr}}

\newcommand{\FV}{\mathsf{FV}}

\newcommand{\Inh}{\mathsf{Inh}}

\newcommand{\ty}{\boldsymbol{\tau}}
\newcommand{\lke}{\mathsf{ke}^\nucleus}
\newcommand{\lif}{\sfif^\nucleus}
\newcommand{\lrec}{\rec^\nucleus}

\newcommand{\baire}{\N^\N}
\newcommand{\val}{\mathsf{val}}
\newcommand{\modu}{\mathsf{mod}}

\newcommand{\cc}{\mathbf{c}}

\begin{document}

\maketitle

\begin{abstract}
  This paper separates two components of functional interpretations: affine information propagation and contraction.
  We introduce information nuclei as an algebraic interface to capture the affine component.
  An information nucleus specifies what information is associated with finite-type objects, how exact objects are compatible with such information, and how information is propagated through functions.
  From any information nucleus we obtain a formula translation and a soundness theorem for affine finite-type arithmetic.
  Extending soundness to finite-type arithmetic with contraction requires one additional ingredient: a formula-indexed contraction structure reducing the challenges generated by duplicated assumptions to a single challenge.
  Finite collections of candidates with union yield a Herbrand-style interpretation, while exact information with challenge selection yields the usual Dialectica interpretation over an arithmetic system restricted to decidable primitive formulas.
  The resulting framework provides a uniform method for specifying the information carried by extracted realizers, allowing existing functional interpretations to be systematically enriched with auxiliary data, such as continuity information.
\end{abstract}

\section{Introduction}
\label{sec:introduction}

Functional interpretations are a central tool for extracting computational information from proofs, going back to G{\"o}del's Dialectica interpretation~\cite{Goedel1958Dialectica}; see also Kohlenbach~\cite{Kohlenbach2008Book} for a modern systematic treatment.
One important application is proof mining, where functional interpretations are used to obtain explicit quantitative data from prima facie nonconstructive proofs.
Many variants of functional interpretation are known, such as the Diller--Nahm~\cite{DillerNahm1974}, monotone~\cite{Kohlenbach1996}, bounded~\cite{FerreiraOliva2005BFI}, and Herbrandized variants~\cite{BBS2012}.
Their presentations often combine several proof-theoretic mechanisms at once: the representation and propagation of witness and challenge information, the treatment of duplicated assumptions, and the handling of extensionality.

The starting point of this paper is that the first two mechanisms can be separated.
In an affine setting, where assumptions are not duplicated, soundness of a functional interpretation only requires operations that propagate witness and challenge information along the logical and arithmetic rules.
Once contraction is added, one must also explain how the challenge information produced by two uses of the same assumption is replaced by a single one.
The third issue, extensionality, is deliberately left aside in this paper, so that the separation between affine propagation and contraction can be treated without further complications.

This perspective is close in spirit to unifying, parametrized, and uniform accounts of
functional interpretations~\cite{Oliva2006Unifying,Oliva2014PastFuture,DinisOliva2021Parametrised,Oliva2025Uniform}, but the emphasis is different.
Rather than seeking a single framework that recovers many interpretations, we isolate a small algebraic interface for the affine core of functional interpretations and then ask what extra structure is needed for handling contraction.
This modular viewpoint separates the specification of the information carried by extracted realizers from the additional structure needed to combine information arising from duplicated assumptions.

This paper introduces \emph{information nuclei} as such an algebraic interface.  An information nucleus assigns to each finite type a type of information objects, a compatibility relation between exact objects and information objects, and basic operations for producing and propagating information.  From these data, we can define the realizer and challenge types, translate formulas, and prove soundness for affine finite-type arithmetic.  The passage to soundness with contraction is then treated separately: once the nucleus is equipped with formula-indexed \emph{contraction operations} that reduce the challenges generated by duplicated assumptions to a single challenge, the affine soundness proof extends to the full system.

The main technical results are therefore two soundness theorems.
First, any information nucleus gives a functional interpretation of an affine finite-type arithmetic, where assumptions may be weakened but not contracted.
Second, if the nucleus is equipped with a suitable contraction structure, the interpretation extends to finite-type arithmetic with contraction.
For example, representing information by finite collections of possible values and interpreting contraction by union gives a Herbrand-style interpretation, while representing information by exact values and interpreting contraction by challenge selection gives the usual Dialectica interpretation when the primitive formulas of the source system are decidable.
Moreover, the choice of information nucleus determines the information carried by extracted realizers and can systematically enrich it with auxiliary structure; continuity nuclei from the parametrized translation of System~\(\T\)~\cite[Section~3.3]{Xu2020Gentzen} give examples in which extracted realizers carry continuity data.

The paper is organized as follows.  Section~\ref{sec:background} fixes the finite-type arithmetic systems used in the paper, with the complete description given in Appendix~\ref{app:systems}.  Section~\ref{sec:interface} introduces information nuclei and the induced formula translation.  Section~\ref{sec:affine-soundness} proves affine soundness.  Section~\ref{sec:contraction} adds contraction structures and proves full soundness.  Section~\ref{sec:related} discusses related work.  The detailed soundness proofs are given in Appendix~\ref{app:proofs}, and Appendix~\ref{app:continuity-fi} spells out the continuity example.

\section{Background: arithmetic systems in finite types}
\label{sec:background}

This section fixes the finite-type arithmetic systems used in the paper: the affine neutral system \(\NHAoa\), its extension \(\NHAo\) by contraction, and an arithmetic system \(\OHAo\) with observational equality.
Full sequent rules and arithmetic axioms are listed in Appendix~\ref{app:systems}.

\subsection{The finite-type term language}
\label{sec:term-language}

We work with the usual finite types over the base type \(\N\), but include the Boolean type \(\Two\) and the unit type \(\One\) as convenient primitive types.  Thus types are generated by
\[
  \sigma,\tau ::= \N \mid \Two \mid \One \mid \sigma\times\tau \mid \sigma\to\tau .
\]
The additional base types \(\Two\) and \(\One\) are inessential: they could be encoded in \(\N\), but keeping them primitive simplifies the presentation.

We use Gödel's system \(\T\) as the underlying term language: terms are built from variables, lambda abstraction and application, the constants
\[
  0:\N,\qquad \suc:\N\to\N,\qquad \sft,\sff:\Two,\qquad \star:\One,
\]
and the product operations
\[
  \pair{\cdot}{\cdot}_{\sigma,\tau} : \sigma \to \tau \to \sigma\times\tau,
  \qquad
  \pr_1^{\sigma,\tau} : \sigma\times\tau \to \sigma,
  \qquad \text{and} \qquad
  \pr_2^{\sigma,\tau} : \sigma\times\tau \to \tau .
\]
In addition, for every type \(\sigma\), we have a conditional and a primitive recursor:
\[
  \mathsf{if}_\sigma : \Two \to \sigma \to \sigma \to \sigma
  \qquad \text{and} \qquad
  \mathsf{rec}_\sigma : \sigma \to (\N \to \sigma \to \sigma) \to \N \to \sigma .
\]
The defining equations for these term formers are standard and are listed in Appendix~\ref{app:systems}.

\medskip
\noindent\textbf{Notation.}
Throughout the paper, we use \(\eqdef\) for definitions and \(\equiv\) for meta-level notational equality between expressions.
We omit type subscripts when they are clear from context.
Finite products are encoded by a fixed association of binary products; finite tuples are written with angle brackets and their components by \(w_i\).
Multi-argument functions are curried; for readability, \(f(a,b,c)\) abbreviates \(f\,a\,b\,c\), not application to a triple.
Function composition is written \(g\circ f\), that is, \(g\circ f \eqdef \lambda x. g(fx)\).

\subsection{Affine neutral arithmetic \texorpdfstring{\(\NHAoa\)}{N-HA omega affine}}

We now introduce the affine neutral arithmetic system \(\NHAoa\). Formulas are generated by
\[
  A,B ::= \bot \mid s =_\sigma t \mid A\wedge B \mid A\vee B \mid A\to B \mid \forall x^\sigma A \mid \exists x^\sigma A
\]
where \(s\) and \(t\) are terms of type \(\sigma\). Equality is primitive at every finite type, and is neutral in the sense that equality at function type is not identified with pointwise equality. Throughout the paper, \(\top\) abbreviates \(\star =_\One \star\).

The system is formulated as an \emph{intuitionistic single-conclusion sequent calculus}.
It contains the usual single-conclusion left and right logical rules in affine, context-splitting form, together with weakening, but not contraction.
Thus contexts are treated \emph{affinely}: assumptions may be discarded, but a single occurrence of an assumption is not freely reusable by a structural rule.  Since contexts are finite multisets, exchange is built into the representation of contexts rather than included as a rule.

The \emph{equality axioms} say that each \(=_{\sigma}\) is an equivalence relation. They also include congruence for term contexts: for every term \(r:\tau\) and variable \(x:\sigma\),
\[
  s =_\sigma t \to r[x \eqsubst s] =_\tau r[x \eqsubst t].
\]
Consequently, equality is substitutive for formulas: from \(s=_\sigma t\) and
\(A(s)\), one can derive~\(A(t)\).

\emph{Induction} is included as a closed rule:
\[
  \inferrule*[right=\(\mathrm{ind}\).]
  {\vdash A(0) \\
   \vdash \forall n^\N(A(n)\to A(\suc n))}
  {\vdash \forall n^\N A(n)}
\]
This formulation is affine: the premises are closed, so no assumption from a surrounding context is reused along the induction.  In contrast, the usual induction axiom, or equivalently an induction rule with arbitrary context, packages the induction data as reusable assumptions and therefore relies on contraction-like behaviour.

The full list of sequent rules and arithmetic axioms is given in Appendix~\ref{app:systems}.

\subsection{Contraction and full neutral arithmetic \texorpdfstring{\(\NHAo\)}{N-HA omega}}

The full neutral system \(\NHAo\) is obtained from \(\NHAoa\) by adding the contraction rule
\[
  \inferrule*[right=\(\mathrm{contr}\)]
  {\Gamma,A,A \vdash B}
  {\Gamma,A \vdash B}.
\]
All other term-forming operations, formulas, equality axioms, and the closed induction rule are the same as in \(\NHAoa\).
In particular, once contraction is available, the usual induction axiom, or equivalently an induction rule with arbitrary context, is derivable from the closed induction rule.


\subsection{Observational equality and extensionality}

We also use an arithmetic system with \emph{observational equality}, denoted \(\OHAo\).
It has the same finite-type term language as \(\NHAo\), full intuitionistic rules with contraction, and the same induction principle, but differs in its treatment of equality.
Its primitive atomic formulas are equalities only at the base types \(\N\), \(\Two\), and \(\One\), and these base-type equalities are decidable by Boolean terms, i.e., for each base type \(\iota\), there is a closed term \(\mathsf{eq}_\iota : \iota \to \iota \to \Two \) such that \(\OHAo\) proves \( \mathsf{eq}_\iota(x,y)=t \leftrightarrow x=_\iota y \).
Consequently, every quantifier-free formula built from the primitive equalities is likewise decided by a Boolean term.
Equality at product and function types is not primitive; it is defined recursively in terms of lower-type equality:
\[
  x =_{\sigma\times\tau} y \eqdef
  \pr_1(x) =_\sigma \pr_1(y) \wedge \pr_2(x) =_\tau \pr_2(y)
  \quad \text{and} \quad
  f =_{\sigma\to\tau} g \eqdef
  \forall x^\sigma\, (fx =_\tau gx).
\]
The product clause is included for uniformity; by the \(\eta\)-rule for products, it is provably equivalent to the usual componentwise equality on pairs.
Only primitive equalities support substitution in arbitrary terms. Equality at product and function types is merely defined, so observational equality does not by itself enjoy congruence.

The standard weakly extensional system \(\WEHAo\) is obtained from \(\OHAo\) by by adding the quantifier-free rule of extensionality, while the fully extensional system \(\EHAo\) is obtained by adding full extensionality (see, e.g.,~\cite{TroelstraVanDalen1988,Kohlenbach2008Book}).
We mention these systems only to locate \(\OHAo\) among the standard extensional variants of finite-type arithmetic systems.
Weak and full \emph{extensionality are not treated} by the general soundness theorems of this paper.

\subsection{Comparison of the systems}

The systems used in the paper differ along two independent axes: whether contraction is available, and how equality at higher types is treated.
They are summarized in Table~\ref{tab:systems}.

\begin{table}[htbp]
  \centering
  \caption{Finite-type arithmetic systems used in the paper.}
  \label{tab:systems}
  \begin{tabular}{@{}llll@{}}
    \toprule
    System     & Contraction & Higher-type equality    & Role in the paper            \\
    \midrule
    \(\NHAoa\) & no          & primitive, neutral      & affine soundness             \\
    \(\NHAo\)  & yes         & primitive, neutral      & soundness with contraction   \\
    \(\OHAo\)  & yes         & defined observationally & standard Dialectica instance \\
    \bottomrule
  \end{tabular}
\end{table}

In \(\NHAoa\) and \(\NHAo\), equality is primitive at every type and therefore enjoys the built-in congruence (substitution) rules of the logic. By contrast, \(\OHAo\) has primitive equality only at the base types; equality at higher types is defined observationally and does not by itself support substitution in arbitrary terms.

We include \(\OHAo\) in the paper mainly to recover the standard Dialectica interpretation in our framework.
In that case, counterexample selection is used to interpret contraction; the argument ultimately relies on the decidability of the primitive base-type formulas of \(\OHAo\).

\section{An algebraic interface for information propagation}
\label{sec:interface}

This section introduces the algebraic interface used by the interpretation,
packaged as a structure called an \emph{information nucleus}.  An information
nucleus specifies how information about finite-type objects is represented, how
exact objects give rise to information, and how information is propagated
through terms.  From this structure we define the realizer and challenge types,
the translated formulas, and the basic realizer operations used later in the
soundness proof.  The structural problem of contraction is deliberately left
out and will be treated separately in Section~\ref{sec:contraction}.

Throughout this section, \(\sys\) denotes a fixed target system, always assumed to be \(\NHAoa\) or an extension of \(\NHAoa\).
An information nucleus is always taken over this target system: its compatibility relations are \(\sys\)-formulas, and its defining properties are required to be provable in \(\sys\).
The formula translation sends source formulas to \(\sys\)-formulas, and the soundness theorems extract terms whose correctness is derivable in \(\sys\).

\subsection{Information nuclei}
\label{sec:nuclei}

We now give the formal definition of an information nucleus. It
has both syntactic and logical components: it specifies the information types
and operations used in the extracted terms, as well as the compatibility
formulas and axioms used to justify their correctness.

\begin{definition}[Information nucleus]
  \label{def:information-nucleus}
  An \emph{information nucleus over \(\sys\)} consists of:
  \begin{itemize}
    \item for each finite type \(\sigma\), a type \(\sigma^\nucleus\);
    \item for each finite type \(\sigma\), an \(\sys\)-formula \(x \compat_\sigma a\), where \(x:\sigma\) and \(a:\sigma^\nucleus\);
    \item for each finite type \(\sigma\), a term \(\eta_\sigma:\sigma\to\sigma^\nucleus\);
    \item for each pair of finite types \(\sigma\) and \(\tau\), a term \(\kappa_{\sigma,\tau}: (\sigma\to\tau^\nucleus)\to\sigma^\nucleus\to\tau^\nucleus\).
  \end{itemize}
  We write \(f^\kappa(a)\) for \(\kappa_{\sigma,\tau}(f,a)\).  For \(a,b:\sigma^\nucleus\), define
  \(
  a \nule_\sigma b
  \eqdef
  \forall x^\sigma\,(x \compat_\sigma a \to x \compat_\sigma b).
  \)
  The data are required to satisfy, provably in \(\sys\),
  \[
    \begin{aligned}
      \textup{(IN1)}\quad &
      x \compat_\sigma \eta_\sigma(x), \\
      \textup{(IN2)}\quad &
      x \compat_\sigma a
      \to f(x) \nule_\tau f^\kappa(a),
    \end{aligned}
  \]
  where the variables \(x:\sigma\), \(a:\sigma^\nucleus\), and \(f:\sigma\to\tau^\nucleus\) are universally quantified within \(\sys\).
\end{definition}

The elements of \(\sigma^\nucleus\), called \emph{information objects},
represent information associated with objects of type \(\sigma\), without
requiring this information to determine a unique exact object.  We read
\(x \compat_\sigma a\) as saying that the exact object \(x:\sigma\) is
\emph{compatible} with the information object \(a:\sigma^\nucleus\).  For
\(a,b:\sigma^\nucleus\), the relation \(a \nule_\sigma b\) says that every
exact object compatible with \(a\) is also compatible with \(b\). Thus \(b\)
is a weaker, or more permissive, information object than \(a\); equivalently,
\(a\) is at least as precise as \(b\). It is immediate that \(\nule_\sigma\)
is reflexive and transitive.

The components of an information nucleus play different roles in the
interpretation.  The assignment \(\sigma \mapsto \sigma^\nucleus\) supplies the
information types used in the realizer and challenge types of formulas.  The
relation \(x \compat_\sigma a\) gives the bounded quantifiers appearing in the
translated formulas.  The operations \(\eta\) and \(\kappa\) are used in the
extraction of realizing information from proofs: \(\eta\) produces exact
information, while \(\kappa\) propagates information through functions.  The
axioms (IN1) and (IN2) are used to justify the correctness of this extracted
information.

We work relative to an information nucleus over an ambient system \(\sys\). In what follows, type subscripts may be omitted whenever they are clear from context.

We use the compatibility relation to express when an information object is nonempty.

\begin{definition}[Inhabited information]
  \label{def:inhabited-information}
  For \(a:\sigma^\nucleus\), define
  \[
    \Inh_\sigma(a) \eqdef \exists x^\sigma \, x \compat_\sigma a .
  \]
  We say that \(a\) is \emph{inhabited} if \(\Inh_\sigma(a)\) is provable in \(\sys\).
\end{definition}

By (IN1), every exact object \(x:\sigma\) gives an inhabited information object \(\eta_\sigma(x)\).

The compatibility relation determines the bounded quantifiers used in the translation: for a formula \(A(x)\), we write
\[
  \forall x \compat a\, A(x) \eqdef \forall x\,(x \compat a \to A(x))
  \qquad\text{and}\qquad
  \exists x \compat a\, A(x) \eqdef \exists x\,(x \compat a \wedge A(x)).
\]
The preorder \(\nule_\sigma\) interacts with bounded quantification in the expected variance directions:
\begin{lemma}[Bounded monotonicity]
  \label{lem:bounded-monotonicity}
  Let \(a,b:\sigma^\nucleus\) and suppose \(a \nule_\sigma b\). Then
  \[
    \forall x \compat b\, A(x)
    \to
    \forall x \compat a\, A(x)
    \qquad\text{and}\qquad
    \exists x \compat a\, A(x)
    \to
    \exists x \compat b\, A(x).
  \]
\end{lemma}

The definition of information nucleus is deliberately abstract, but the following two instances show how it captures both exact and finite-candidate information.

\begin{example}[Exact information]
  \label{ex:exact-information}
  Take \(\sigma^\nucleus \eqdef \sigma\), let \(x\compat_\sigma a\) mean \(x=_\sigma a\), and set \(\eta_\sigma(x)\eqdef x\) and \(f^\kappa(a)\eqdef f(a)\).
  Then (IN1) is reflexivity of equality, and (IN2) follows from the substitution property of equality.
\end{example}

\begin{example}[Finite candidate information]
  \label{ex:finite-candidate-information}
  A second instance takes information objects to be finite collections of candidates.
  This may be formulated using finite collection types, or by fixing an encoding of finite collections in the finite-type language.
  Writing \(\sigma^*\) for such a finite-collection type or coding, define
  \[
    \sigma^\nucleus \eqdef \sigma^*,
    \qquad
    x\compat_\sigma a \eqdef x\in_\sigma a,
    \qquad
    \eta_\sigma(x)\eqdef [x],
    \qquad
    f^\kappa(a)\eqdef \bigcup_{x\in a} f(x).
  \]
  Here \(x\in_\sigma a\) denotes membership, \([x]\) is the singleton collection, and \(\bigcup_{x\in a} f(x)\) is the finite union, or concatenation, of the collections \(f(x)\) for \(x\) in \(a\).
  Then (IN1) says that \(x\in [x]\), and (IN2) says that if \(x\in a\), then \(f(x)\) is included in the union of the \(f(y)\) for \(y\in a\).
  This is the finite-candidate pattern underlying Herbrand-style interpretations.
\end{example}

The preceding examples recover familiar proof-interpretation patterns; the next
one illustrates that information nuclei can also enrich extracted realizers with
additional computational structure.

\begin{example}[Continuity information]
  \label{ex:continuity-information}
  Continuity nuclei from the parametrized translation of System~\(\T\)
  \cite[Section~3.3]{Xu2020Gentzen} give examples in which realizers carry
  continuity data.

  In the case of pointwise continuity, an information object for a type \(\sigma\) consists of a value component \(v: \baire\to\sigma\) and a modulus component \(m: \baire\to\N\).
  Compatibility is defined relative to an oracle \(\alpha:\baire\), written \(\compat^\alpha_\sigma\):
  \[
    x \compat^\alpha_\sigma (v,m) \eqdef x =_\sigma v(\alpha) \wedge \text{\(m(\alpha)\) witnesses the local constancy of \(v\) at \(\alpha\).}
  \]
  The unit \(\eta\) maps an exact object to the information object with constant value function and contant-zero modulus.
  The Kleisli extension \(\kappa\) composes the value components and
  propagates the continuity information by combining the modulus required for
  the input with the modulus required for the output computation.  The precise
  oracle-parametrized compatibility relation, the definitions of these
  operations, and the verification of (IN1) and (IN2) are given in
  Appendix~\ref{app:continuity-fi}.
  A uniform-continuity variant can be obtained similarly, by replacing pointwise moduli with uniform modulus information.
\end{example}

\subsection{Realizers, challenges, and propagated information}
\label{sec:realizers}

We now fix a source system \(\src\), taken to be one of \(\NHAoa\), \(\NHAo\), or \(\OHAo\), while the target system \(\sys\) remains fixed by the target convention above.
The interpretation is carried out in~\(\sys\), relative to the fixed information nucleus \(((\text{-})^\nucleus,\compat,\eta,\kappa)\) over \(\sys\).

Each formula \(A\) of the source system \(\src\) is assigned two types in the target system \(\sys\): a positive type \(\ty_A^+\) of realizers and a negative type \(\ty_A^-\) of challenges.  In the present setting, realizers are understood broadly: they may contain information about witnesses or choices, rather than exact objects themselves. These types are defined recursively from the logical form of \(A\), using the information types supplied by the fixed nucleus. Here and below, \(P\)~ranges over primitive formulas of \(\src\), including \(\bot\) and the atomic formulas of \(\src\).

\begin{definition}[Positive and negative types]
  \label{def:positive-negative-types}
  For each formula \(A\) of the source system \(\src\), define types
  \(\ty_A^+\) and \(\ty_A^-\) by recursion on \(A\):
  \[
    \begin{array}{rclcrcl}
      \ty^+_P           & \eqdef & \One
                        & \qquad &
      \ty^-_P           & \eqdef & \One
      \\[0.3em]

      \ty^+_{B\wedge C} & \eqdef & \ty^+_B\times\ty^+_C
                        & \qquad &
      \ty^-_{B\wedge C} & \eqdef & \ty^-_B\times\ty^-_C
      \\[0.3em]

      \ty^+_{B\vee C}
                        & \eqdef &
      \Two^\nucleus\times\ty^+_B\times\ty^+_C
                        &
                        &
      \ty^-_{B\vee C}
                        & \eqdef &
      (\ty^-_B)^\nucleus\times(\ty^-_C)^\nucleus
      \\[0.3em]

      \ty^+_{B\to C}
                        & \eqdef &
      (\ty^+_B\to\ty^+_C)
      \times
      (\ty^+_B\times\ty^-_C\to(\ty^-_B)^\nucleus)
                        &
                        &
      \ty^-_{B\to C}
                        & \eqdef &
      \ty^+_B\times\ty^-_C
      \\[0.3em]

      \ty^+_{\forall x^\sigma B}
                        & \eqdef &
      \sigma^\nucleus\to\ty^+_B
                        &
                        &
      \ty^-_{\forall x^\sigma B}
                        & \eqdef &
      \sigma^\nucleus\times\ty^-_B
      \\[0.3em]

      \ty^+_{\exists x^\sigma B}
                        & \eqdef &
      \sigma^\nucleus\times\ty^+_B
                        &
                        &
      \ty^-_{\exists x^\sigma B}
                        & \eqdef &
      (\ty^-_B)^\nucleus .
    \end{array}
  \]
\end{definition}

In these clauses, information types enter exactly where the interpretation uses information rather than exact data: \(\Two^\nucleus\) for the chosen disjunct, \(\sigma^\nucleus\) for quantified objects, and \((\ty_B^-)^\nucleus\) for bounds on challenges.

We next lift the information preorder from information objects to realizers of
formulas.  Recall that \(a \nule_\sigma b\) means that \(a\) is at least as
precise as \(b\).  For each formula \(A\), we define a corresponding
preorder \(r \nule_A s\) on realizers \(r,s:\ty_A^+\).

\begin{definition}[Preorder on realizers]
  \label{def:realizer-preorder}
  For each formula \(A\) of \(\src\), define an \(\sys\)-formula
  \(r \nule_A s\), where \(r,s:\ty_A^+\), by recursion on \(A\):
  \[
    \begin{array}{rcl}
      r \nule_P s
       & \eqdef &
      \top,
      \\[0.3em]

      r \nule_{B\wedge C} s
       & \eqdef &
      r_1 \nule_B s_1 \wedge r_2 \nule_C s_2,
      \\[0.3em]

      r \nule_{B\vee C} s
       & \eqdef &
      r_1 \nule_{\Two} s_1
      \wedge r_2 \nule_B s_2
      \wedge r_3 \nule_C s_3,
      \\[0.3em]

      r \nule_{B\to C} s
       & \eqdef &
      \bigl(\forall x^{\ty_B^+}\, r_1(x) \nule_C s_1(x)\bigr)
      \wedge
      \bigl(\forall w^{\ty_B^+\times\ty_C^-}\,
      r_2(w) \nule_{\ty_B^-} s_2(w)\bigr),
      \\[0.3em]

      r \nule_{\forall x^\sigma B} s
       & \eqdef &
      \forall a^{\sigma^\nucleus}\, r(a) \nule_B s(a),
      \\[0.3em]

      r \nule_{\exists x^\sigma B} s
       & \eqdef &
      r_1 \nule_\sigma s_1 \wedge r_2 \nule_B s_2 .
    \end{array}
  \]
\end{definition}

Thus \(r \nule_A s\) means that \(r\) is at least as precise as \(s\) as a
realizer of \(A\).  It is immediate by induction on \(A\) that \(\nule_A\) is
reflexive and transitive.

We now lift the Kleisli extension \(\kappa\) of the nucleus from information objects to
realizers of formulas.  Given a family of realizers \(f:\sigma\to\ty_A^+\), the
lifted extension turns information \(a:\sigma^\nucleus\) about an input into a
realizer of \(A\).  The resulting realizer contains enough information to cover
all \(f(x)\) with \(x \compat_\sigma a\).

\begin{definition}[Kleisli extension to realizers]
  \label{def:lke}
  For each finite type \(\sigma\) and formula \(A\) of \(\src\), define a term
  \[
    \lke_{\sigma,A} : (\sigma \to \ty_A^+) \to \sigma^\nucleus \to \ty_A^+
  \]
  by recursion on \(A\):
  \[
    \begin{array}{rcl}
      \lke_{\sigma,P}(f,a)
       & \eqdef &
      \star,
      \\[0.3em]

      \lke_{\sigma,B\wedge C}(f,a)
       & \eqdef &
      \langle
      \lke_{\sigma,B}(\lambda x.\,(f(x))_1,\,a),\;
      \lke_{\sigma,C}(\lambda x.\,(f(x))_2,\,a)
      \rangle,
      \\[0.3em]

      \lke_{\sigma,B\vee C}(f,a)
       & \eqdef &
      \langle
      (\lambda x.\,(f(x))_1)^\kappa(a),\;
      \lke_{\sigma,B}(\lambda x.\,(f(x))_2,\,a),\;
      \lke_{\sigma,C}(\lambda x.\,(f(x))_3,\,a)
      \rangle,
      \\[0.3em]

      \lke_{\sigma,B\to C}(f,a)
       & \eqdef &
      \langle
      \lambda r^{\ty_B^+}.\, \lke_{\sigma,C}(\lambda x.\,(f(x))_1(r),\,a),\;
      \lambda w^{\ty_B^+\times\ty_C^-}.\, (\lambda x.\,(f(x))_2(w))^\kappa(a)
      \rangle,
      \\[0.3em]

      \lke_{\sigma,\forall y^\tau B}(f,a)
       & \eqdef &
      \lambda b^{\tau^\nucleus}.\,
      \lke_{\sigma,B}(\lambda x.\,f(x,b),\,a),
      \\[0.3em]

      \lke_{\sigma,\exists y^\tau B}(f,a)
       & \eqdef &
      \langle
      (\lambda x.\,(f(x))_1)^\kappa(a),\;
      \lke_{\sigma,B}(\lambda x.\,(f(x))_2,\,a)
      \rangle .
    \end{array}
  \]
\end{definition}

The lifted Kleisli extension satisfies the realizer-level analogue of (IN2): information about an input is enough to propagate a family of realizers indexed by exact inputs.

\begin{lemma}[Realizer-level propagation]
  \label{lem:lke-propagation}
  For every finite type \(\sigma\) and every formula \(A\) of \(\src\), \(\sys\)
  proves
  \[
    x \compat_\sigma a
    \to
    f(x) \nule_A \lke_{\sigma,A}(f,a)
  \]
  for all \(x:\sigma\), \(a:\sigma^\nucleus\), and
  \(f:\sigma\to\ty_A^+\).
\end{lemma}

\begin{proof}
  By induction on \(A\).  The primitive case is immediate.  The conjunction,
  disjunction, implication, universal, and existential cases follow from the
  induction hypotheses, using (IN2) in the components where the
  definition of \(\lke\) applies \(\kappa\).
\end{proof}

The lifted Kleisli extension yields two derived operations that will be used later in the soundness proof: a lifted Boolean case operator and a lifted primitive recursor.

\begin{definition}[Lifted case distinction]
  \label{def:lifted-if}
  For each formula \(A\) of \(\src\), define
  \[
    \begin{aligned}
      \lif_A & : \Two^\nucleus\to\ty_A^+\to\ty_A^+\to\ty_A^+ \\
      \lif_A(p,r,s)
             & \eqdef
      \lke_{\Two,A}
      (\lambda b^\Two.\,\mathsf{if}_{\ty_A^+}(b,r,s),\, p).
    \end{aligned}
  \]
\end{definition}

\begin{definition}[Lifted primitive recursion]
  \label{def:lifted-rec}
  For each formula \(A\) of \(\src\), define
  \[
    \begin{aligned}
      \lrec_A & : \ty_A^+ \to (\N^\nucleus \to \ty_A^+ \to \ty_A^+) \to \N^\nucleus \to \ty_A^+ \\
      \lrec_A(r,f,p)
              & \eqdef
      \lke_{\N,A}
      (\lambda n^\N.\,\mathsf{rec}_{\ty_A^+}(r,f \circ \eta,n),\;p).
    \end{aligned}
  \]
\end{definition}

The correctness of the above operations is immediate from Lemma~\ref{lem:lke-propagation}.
That is, the lifted operations produce information objects compatible with the results of the corresponding exact operations.

\begin{corollary}[Correctness of lifted case distinction]
  \label{cor:lifted-if}
  For every formula \(A\) of \(\src\), \(\sys\) proves
  \[
    \sft \compat p \to r \nule_A \lif_A(p,r,s)
    \qquad\text{and}\qquad
    \sff \compat p \to s \nule_A \lif_A(p,r,s).
  \]
\end{corollary}

\begin{corollary}[Correctness of lifted recursion]
  \label{cor:lifted-rec}
  For every formula \(A\) of \(\src\), \(\sys\) proves
  \[
    n \compat_\N p
    \to
    \mathsf{rec}_{\ty_A^+}(r, f \circ \eta,n)
    \nule_A
    \lrec_A(r,f,p)
  \]
  for all \(n:\N\), \(p:\N^\nucleus\), \(r:\ty_A^+\), and \(f:\N^\nucleus\to\ty_A^+\to\ty_A^+\).
\end{corollary}

\subsection{The formula translation induced by an information nucleus}
\label{sec:formulas}

We now define the translated formula.  For each source formula \(A\), a
realizer \(r:\ty_A^+\) is tested against a challenge \(u:\ty_A^-\), producing
an \(\sys\)-formula \(|A|^r_u\).  An information object
\(a:(\ty_A^-)^\nucleus\) is called a \emph{challenge bound} for \(A\): it
represents information about possible challenges to \(A\).  Thus the bounded
formula \(\forall z\compat a\,|A|^r_z\) says that \(r\) meets every exact
challenge \(z\) compatible with the bound \(a\).  The clauses below follow the
logical form of \(A\), while the information nucleus controls how witnesses,
choices, and challenges are represented by information objects.

\begin{definition}[Formula translation]
  \label{def:formula-translation}
  For each formula \(A\) of \(\src\), define an \(\sys\)-formula \(|A|^r_u\),
  where \(r:\ty_A^+\) and \(u:\ty_A^-\), by recursion on \(A\):
  \[
    \begin{array}{rcl}
      |P|^r_u
       & \eqdef &
      P,
      \\[0.3em]

      |B\wedge C|^r_u
       & \eqdef &
      |B|^{r_1}_{u_1}
      \wedge
      |C|^{r_2}_{u_2},
      \\[0.3em]

      |B\vee C|^r_u
       & \eqdef &
      \bigl(\sft \compat r_1 \wedge \forall v \compat u_1\, |B|^{r_2}_v\bigr)
      \vee
      \bigl(\sff \compat r_1 \wedge \forall v \compat u_2\, |C|^{r_3}_v\bigr),
      \\[0.3em]

      |B\to C|^r_u
       & \eqdef &
      \bigl(\forall v \compat r_2(u)\, |B|^{u_1}_v\bigr)
      \to
      |C|^{r_1(u_1)}_{u_2},
      \\[0.3em]

      |\forall x^\sigma B(x)|^r_u
       & \eqdef &
      \forall x \compat u_1\, |B(x)|^{r(u_1)}_{u_2},
      \\[0.3em]

      |\exists x^\sigma B(x)|^r_u
       & \eqdef &
      \exists x \compat r_1\, \forall v \compat u\, |B(x)|^{r_2}_v .
    \end{array}
  \]
\end{definition}

The disjunction clause uses information \(r_1:\Two^\nucleus\) about the chosen side.
The universal clause is challenged by information \(u_1:\sigma^\nucleus\) about the quantified object, while the existential clause realizes the witness by information \(r_1:\sigma^\nucleus\).
In the implication and existential clauses, challenges to the antecedent or body are bounded by information objects.

The preorder on realizers is compatible with the translation: replacing a
realizer by a less precise one preserves validity against the same
challenge.

\begin{lemma}[Monotonicity in realizers]
  \label{lem:formula-monotonicity}
  For every formula \(A\) of \(\src\), \(\sys\) proves
  \[
    |A|^r_u \wedge r \nule_A s
    \to
    |A|^s_u
  \]
  for all \(r,s:\ty_A^+\) and \(u:\ty_A^-\).
\end{lemma}

\begin{proof}
  By induction on \(A\).
  The primitive case is immediate.
  The other cases follow from the induction hypotheses and bounded monotonicity (Lemma~\ref{lem:bounded-monotonicity}).
\end{proof}

\section{Affine soundness}
\label{sec:affine-soundness}

We now prove the first payoff of the interface developed in Section~\ref{sec:interface}: an information nucleus suffices to interpret the affine system \(\NHAoa\).
The formula translation and the realizer operations defined above use only the nucleus, and the absence of contraction means that no operation is needed for combining the challenge bounds generated by duplicated uses of the same assumption.

\medskip
\noindent\textbf{Notation.}
For a context \(\Gamma \equiv C_1,\ldots,C_n\), variables
\(\vec x \equiv x_1,\ldots,x_n\) with \(x_i:\ty_{C_i}^+\), and challenge
bounds \(\vec u \equiv u_1,\ldots,u_n\) with
\(u_i:(\ty_{C_i}^-)^\nucleus\), we abbreviate the context consisting of translated assumptions by
\[
  \forall \vec z \compat \vec u\, |\Gamma|^{\vec x}_{\vec z}
  \qquad\text{for}\qquad
  \forall z_1 \compat u_1\, |C_1|^{x_1}_{z_1},
  \ldots,
  \forall z_n \compat u_n\, |C_n|^{x_n}_{z_n}.
\]

The formulation of soundness below follows the style of Schwichtenberg and Wainer~\cite[Section~7.4]{SchwichtenbergWainer2011}, which is technically well suited for implementation.

\begin{theorem}[Affine soundness]
  \label{thm:affine-soundness}
  Let the target system \(\sys\) be \(\NHAoa\) or an extension of \(\NHAoa\).
  Let \(((\text{-})^\nucleus,\compat,\eta,\kappa)\) be an information nucleus over \(\sys\).
  If
  \[
    \Gamma \vdash_{\NHAoa} A,
  \]
  where \(\Gamma \equiv C_1,\ldots,C_n\), then, with fresh variables
  \(x_i:\ty_{C_i}^+\) and \(y:\ty_A^-\), one can extract terms
  \[
    r:\ty_A^+
    \qquad\text{and}\qquad
    u_i:(\ty_{C_i}^-)^\nucleus
  \]
  such that, whenever all information-object components occurring in the conclusion challenge \(y:\ty_A^-\) are inhabited, each extracted challenge bound \(u_i\) is inhabited, and
  \[
    \forall \vec z \compat \vec u\, |\Gamma|^{\vec x}_{\vec z}
    \;\vdash_{\sys}
    |A|^r_y .
  \]
  Moreover, writing \(\FV(\text{-})\) for the set of free variables, the extracted terms satisfy
  \[
    \FV(r) \subseteq \FV(\Gamma,A)\cup\{x_1,\ldots,x_n\}
    \quad\text{and}\quad
    \FV(u_i) \subseteq \FV(\Gamma,A)\cup\{x_1,\ldots,x_n,y\}
  \]
  for each \(i=1,\ldots,n\).
  In particular, the challenge variable \(y\) does not occur free in~\(r\).
\end{theorem}

\begin{proof}
  The proof is by induction on the given derivation in \(\NHAoa\).  The complete
  case analysis is given in Appendix~\ref{app:proofs}; here we indicate
  the structure of the argument and the two cases where the operations from
  Section~\ref{sec:interface} are essential.

  The logical rules are handled by the corresponding clauses of the formula translation (Definition~\ref{def:formula-translation}).
  Monotonicity in realizers (Lemma~\ref{lem:formula-monotonicity}) allows us to weaken realizing information, replacing a realizer by a less precise one.
  The lifted Kleisli extension (Definition~\ref{def:lke}) is used to pass from exact inputs to information about inputs; its correctness is expressed by realizer-level propagation (Lemma~\ref{lem:lke-propagation}). The lifted case distinction (Definition~\ref{def:lifted-if}) and the lifted recursion (Definition~\ref{def:lifted-rec}) are used in the cases of disjunction left and closed induction, respectively, as follows.

  \medskip
  \noindent\emph{Disjunction left.}
  Consider the affine disjunction-left rule
  \[
    \inferrule*[right=\(\vee\mathrm{L}\)]
    {\Gamma,A \vdash C
      \\
      \Delta,B \vdash C}
    {\Gamma,\Delta,A\vee B \vdash C}.
  \]
  By the induction hypotheses for the two premises, we obtain realizers \(r:\ty_C^+\) and \(s:\ty_C^+\) for the two conclusions.
  Given a variable \(t:\ty_{A\vee B}^+\) that realizes the assumption \(A\vee B\) in the conclusion, write \(t \equiv \langle t_1,t_2,t_3\rangle\), where
  \(t_1:\Two^\nucleus\) is information about the chosen disjunct,
  \(t_2:\ty_A^+\), and \(t_3:\ty_B^+\).
  We omit the routine challenge bookkeeping and construct the realizer for the conclusion:
  \[
    q \eqdef \lif_C(t_1,r,s) : \ty_C^+ .
  \]
  It remains to verify that \(q\) realizes \(C\) against any challenge \(y:\ty_C^-\).
  The translation of the assumption \(A \vee B\), with realizer \(t\), unfolds as a disjunction.
  In the left case, we have \(\sft \compat t_1\) and the translated assumption for \(A\), exactly as needed to apply the induction hypothesis for the left premise; and we obtain \(|C|^r_y\).
  Since \(\sft \compat t_1\), correctness of lifted case distinction (Corollary~\ref{cor:lifted-if}) gives \(r \nule_C q\); monotonicity in realizers (Lemma~\ref{lem:formula-monotonicity}) then yields \(|C|^q_y\).
  The right case is analogous and gives \(|C|^q_y\) as well.

  \smallskip
  \noindent\emph{Closed induction.}
  Consider the closed induction rule
  \[
    \inferrule*[right=\(\mathrm{ind}\)]
    {\vdash A(0) \\
      \vdash \forall n^\N(A(n)\to A(\suc n))}
    {\vdash \forall n^\N A(n)}.
  \]
  By the induction hypotheses for the two premises, we obtain a realizer
  \(r_0:\ty_A^+\) for the base case and a realizer
  \(r_S:\N^\nucleus\to\ty^+_{A(n)\to A(\suc n)}\) for the step.
  Writing \(r_S(p)\equiv \langle r_S^+(p),r_S^-(p)\rangle\),
  its realizer-transforming component \(r_S^+(p)\) has type \(\ty_A^+\to\ty_A^+\).
  The realizer \(h\) for the conclusion \(\forall n^\N A(n)\) is defined using the lifted primitive recursion:
  \[
    h \eqdef
    \lrec_A\bigl(r_0,\lambda p^{\N^\nucleus}. r_S^+(p)\bigr)
    : \N^\nucleus \to \ty_A^+ .
  \]
  The verification is by closed induction in the target system on the exact number compatible with the input information.
  The full details are given in Appendix~\ref{app:proofs}.
\end{proof}

The inhabitedness condition on extracted challenge bounds is essential for soundness with general information nuclei, where information objects need not be inhabited automatically.
It is stated relative to the conclusion challenge because some rules build assumption bounds by propagating information contained in that challenge; when a challenge variable contains information objects, the proof assumes those components are inhabited at the point where the challenge is used.
Inhabitedness is also needed when a rule recovers assumptions with smaller challenge bounds from an assumption with a combined bound.
For example, in the conjunction-left case, to use a bound for \(A\wedge B\) as a bound for~\(A\), one must fill the unused \(B\)-component with some exact challenge compatible with its bound.
This is possible only when the bound for~\(B\) is inhabited.
The detailed proof in Appendix~\ref{app:proofs} keeps track of this relative inhabitedness requirement in each rule case.

Applying the affine soundness theorem to the exact-information nucleus of Example~\ref{ex:exact-information} gives an affine exact-information interpretation of \(\NHAoa\).
Applying it to the finite-candidate nucleus of Example~\ref{ex:finite-candidate-information} gives an affine finite-candidate interpretation.
These are the affine cores of the Dialectica and Herbrand interpretations, respectively.

Similarly, applying the affine soundness construction to the continuity nucleus
of Example~\ref{ex:continuity-information} gives, uniformly in the
oracle parameter, an affine functional interpretation whose extracted realizers
also carry pointwise-continuity information.  Appendix~\ref{app:continuity-fi}
spells out the nucleus and explains how the continuity information is read off from the extracted realizers.

\begin{remark}
The reason extracted realizers can carry additional information, beyond the exact values extracted via the ordinary exact-information interpretation, is that extraction is compositional in the proof.  The extracted terms are built from the realizer operations associated with the logical and arithmetic rules.  By choosing a nucleus whose information types and operations contain additional data, the same syntactic construction propagates that data together with the ordinary computational content.  The compatibility relation then states the correctness property of the additional components.
\end{remark}

\section{Contraction and full soundness}
\label{sec:contraction}

In the affine soundness theorem, different occurrences of assumptions are interpreted independently.
Thus, if a derivation uses two occurrences of the same formula \(A\), the interpretation may assign them two different challenge bounds \(u,v:(\ty_A^-)^\nucleus\).
This causes no difficulty in the affine system, where the two occurrences remain distinct.
In the full intuitionistic system, however, the contraction rule
\[
  \inferrule*[right=\(\mathrm{contr}\)]
  {\Gamma,A,A \vdash B}
  {\Gamma,A \vdash B}
\]
identifies these two occurrences.
To interpret it, the two bounds generated for the two copies of \(A\) must be replaced by a single bound for the remaining copy.

This section isolates the additional structure needed to interpret this step.
We call such data a \emph{contraction structure} on an information nucleus.
The basic notion of information nucleus describes how information is propagated through terms and formulas; a contraction structure adds the further data needed to account for the structural reuse of assumptions in derivations.
Once such a structure is available, the affine soundness theorem extends to full intuitionistic soundness by adding only the contraction case to the proof.

The contraction case of the soundness proof uses the target system only through the defining property of the contraction operation.
Thus the role of contraction is isolated in this formula-indexed property.
In general, however, verifying this property may itself require contraction in the target system; this is the case for the standard merging and selection examples below.
Therefore, in this section, we take the target system \(\sys\) to be \(\NHAo\) or an extension of it.

\subsection{Contraction structures}
\label{sec:contraction-structures}

We formulate contraction as a formula-indexed operation on challenge bounds.
Given a formula \(A\), the operation takes the two bounds produced for two uses
of \(A\) and returns one bound for a single use.  Its correctness condition says
that validity against the single bound is enough to recover validity against
both original bounds.

\begin{definition}[Contraction structure]
  \label{def:contraction-structure}
  Let \(((\text{-})^\nucleus,\compat,\eta,\kappa)\) be an information nucleus over the target system \(\sys\).
  A \emph{contraction operation} for a formula \(A\) of the source system \(\src\) is a term
  \[
    \ctr_A:
    \ty_A^+ \to
    (\ty_A^-)^\nucleus \to
    (\ty_A^-)^\nucleus \to
    (\ty_A^-)^\nucleus
  \]
  such that \(\sys\) proves
  \[
    \bigl( \forall z \compat \ctr_A(x,u,v)\, |A|^x_z \bigr)
    \, \to \,
    \bigl(
    \forall z \compat u\, |A|^x_z
    \wedge
    \forall z \compat v\, |A|^x_z
    \bigr)
  \]
  for all \(x:\ty_A^+\) and \(u,v:(\ty_A^-)^\nucleus\).
  The operation is also required to preserve inhabitedness of challenge bounds:
  \[
    \Inh_{\ty_A^-}(u)\wedge \Inh_{\ty_A^-}(v)
    \to
    \Inh_{\ty_A^-}(\ctr_A(x,u,v)).
  \]
  A \emph{contraction structure} on the information nucleus over \(\sys\), relative to \(\src\), is a choice of a contraction operation for every formula of \(\src\).
\end{definition}

The following examples demonstrate the two basic ways such operations arise.
The first is contraction by merging: challenge bounds are combined uniformly by
joins.  The second is contraction by selection: exact challenges cannot be
joined, so one of them is selected instead.

\begin{example}[Contraction by merging]
  \label{ex:merge-contraction}
  Suppose that the information nucleus has binary joins on information objects: for each type \(\sigma\), there is a term
  \[
    \sqcup_\sigma : \sigma^\nucleus\to\sigma^\nucleus\to\sigma^\nucleus
  \]
  such that
  \[
    a \nule_\sigma a\sqcup_\sigma b
    \qquad\text{and}\qquad
    b \nule_\sigma a\sqcup_\sigma b.
  \]
  Then every formula \(A\) has a contraction operation given by
  \[
    \ctr_A(x,u,v) \eqdef u\sqcup_{\ty_A^-} v .
  \]
  If \(u\) is inhabited, then \(u\sqcup v\) is inhabited by \(u\nule u\sqcup v\); hence the operation also preserves inhabitedness of challenge bounds.
  The defining property follows from bounded monotonicity, together with conjunction introduction in the target system.
  For the finite-candidate nucleus of Example~\ref{ex:finite-candidate-information}, \(\sqcup\) is finite
  union or concatenation of finite collections.
  This is the contraction mechanism used in both Herbrand and Diller--Nahm interpretations.
\end{example}

\begin{example}[Contraction by selection]
  \label{ex:selection-contraction}
  Consider the exact-information nucleus of Example~\ref{ex:exact-information} with source system \(\OHAo\) and target system \(\NHAo\).
  Because \(\sigma^\nucleus \equiv \sigma\) and \(x\compat a\) is \(x=a\), every bounded quantifier introduced by the translation collapses to substitution. Thus \(\forall z\compat u\, |A|^x_z\) is equivalent to \(|A|^x_u\) in \(\NHAo\).  Moreover, \(|A|^x_u\) is decidable: after unfolding observational equality in \(\OHAo\), the only primitive atomic formulas left are equalities at the base types, which are decidable.
  Therefore, for every formula \(A\) of \(\OHAo\) there is a Boolean term
  \[
    d_A : \ty_A^+ \to \ty_A^- \to \Two
  \]
  such that \(\NHAo\) proves, for all \(x:\ty_A^+\) and \(u:\ty_A^-\),
  \[
    d_A(x,u)=\sft \leftrightarrow |A|^x_u .
  \]
  Define the contraction operation by
  \[
    \ctr_A(x,u,v) \eqdef \sfif_{\ty_A^-}\bigl(d_A(x,u),v,u\bigr).
  \]
  Since information objects are exact objects in this nucleus, every challenge bound is inhabited by reflexivity; hence the inhabitedness requirement is automatic.
  We verify the defining property.
  Assume \(|A|^x_{\ctr_A(x,u,v)}\).
  If \(d_A(x,u)=\sft\), then \(|A|^x_u\) holds
  by the specification of \(d_A\), and \(\ctr_A(x,u,v)=v\), so the assumption
  gives \(|A|^x_v\).  If \(d_A(x,u)=\sff\), then \(\ctr_A(x,u,v)=u\), so the
  assumption gives \(|A|^x_u\). The specification of \(d_A\) then gives \(d_A(x,u)=\sft\), contradicting
  \(d_A(x,u)=\sff\); by ex falso, both conclusions follow.

  This is the usual Dialectica counterexample-selection mechanism: if the first
  challenge \(u\) is already satisfied, test the second challenge \(v\);
  otherwise select \(u\) itself as the counterexample.
\end{example}

\subsection{Full soundness}
\label{sec:full-soundness}

We can now extend affine soundness to source systems with contraction.
In this paper the relevant source systems are \(\NHAo\) and \(\OHAo\).
The statement is the same as Theorem~\ref{thm:affine-soundness}, except that
the source system has contraction and the information nucleus is equipped with
a contraction structure for formulas of the chosen source system.

\begin{theorem}[Full soundness]
  \label{thm:full-soundness}
  Let the source system \(\src\) be \(\NHAo\) or \(\OHAo\), and let the target system \(\sys\) be \(\NHAo\) or an extension of \(\NHAo\).
  Let \(((\text{-})^\nucleus,\compat,\eta,\kappa)\) be an information nucleus over
  \(\sys\), equipped with a contraction structure \(\ctr\) relative to \(\src\).  If
  \[
    \Gamma \vdash_{\src} A,
  \]
  where \(\Gamma \equiv C_1,\ldots,C_n\), then, with fresh variables
  \(x_i:\ty_{C_i}^+\) and \(y:\ty_A^-\), one can extract terms
  \[
    r:\ty_A^+
    \qquad\text{and}\qquad
    u_i:(\ty_{C_i}^-)^\nucleus
  \]
  such that the extracted challenge bounds \(u_i\) satisfy the same relative inhabitedness condition as in Theorem~\ref{thm:affine-soundness}, and
  \[
    \forall \vec z \compat \vec u\, |\Gamma|^{\vec x}_{\vec z}
    \;\vdash_{\sys}
    |A|^r_y .
  \]
  Moreover, the same variable condition as in Theorem~\ref{thm:affine-soundness}
  holds: in particular, the challenge variable \(y\) does not occur free in the
  extracted realizer \(r\).
\end{theorem}

\begin{proof}
  The proof is almost the same as the affine soundness proof, with one additional case for
  contraction.  Suppose the last applied rule is
  \[
    \inferrule*[right=\(\mathrm{contr}\)]
    {\Gamma,A,A \vdash B}
    {\Gamma,A \vdash B}.
  \]
  The induction hypothesis for the premise gives two challenge bounds for the
  two occurrences of \(A\), say \(u,v:(\ty_A^-)^\nucleus\), together with a
  derivation of the translated conclusion from assumptions containing
  \[
    \forall z\compat u\, |A|^x_z
    \qquad\text{and}\qquad
    \forall z\compat v\, |A|^x_z .
  \]
  In the conclusion, the single occurrence of \(A\) is assigned the bound
  \(\ctr_A(x,u,v)\).
  This bound is inhabited by the inhabitedness preservation property of \(\ctr_A\), since \(u\) and \(v\) are inhabited by the induction hypothesis.
  The defining property of the contraction operation gives
  \[
    \bigl( \forall z\compat \ctr_A(x,u,v)\, |A|^x_z \bigr)
    \to
    \bigl(
    \forall z\compat u\, |A|^x_z
    \wedge
    \forall z\compat v\, |A|^x_z
    \bigr),
  \]
  so the two assumptions required by the induction hypothesis are recovered
  from the single contracted assumption.

  For \(\OHAo\), equality at product and function types is unfolded according to the observational clauses; the equality cases are discussed in Appendix~\ref{app:proofs}.
\end{proof}

Thus the additional proof-theoretic content needed to pass from affine to full soundness is isolated in the contraction structure.
Therefore, any information nucleus equipped with a contraction structure for the chosen source system yields a functional interpretation of that source system.
For instance, the finite-candidate nucleus of Example~\ref{ex:finite-candidate-information}, together with the merging operation of Example~\ref{ex:merge-contraction}, gives the Herbrand interpretation over \(\NHAo\); and the exact-information nucleus of Example~\ref{ex:exact-information}, together with the selection operation of Example~\ref{ex:selection-contraction}, gives the usual Dialectica interpretation over \(\OHAo\).

\section{Related work}
\label{sec:related}

The closest surrounding literature is the programme of unifying, uniform, and parametrized functional interpretations developed by Oliva and collaborators \cite{Oliva2006Unifying,Oliva2014PastFuture,DinisOliva2021Parametrised,Oliva2025Uniform,FerreiraOliva2025InformativeTypes}.
These frameworks recover many known interpretations by suitable choices of parameters, including Gödel's Dialectica interpretation, the Diller--Nahm variant, modified realizability, monotone and bounded interpretations, and Herbrandized variants.
Our aim is different: instead of seeking maximal coverage, we isolate a minimal algebraic interface sufficient for affine soundness, and then identify the extra structure needed for contraction.

Our framework is particularly close to the unifying functional interpretation of Oliva~\cite{Oliva2006Unifying}.
There, the bounded-quantifier abbreviation is required to satisfy three conditions \(B_1\)--\(B_3\) \cite[Definition~12]{Oliva2006Unifying}.
In our terminology, \(B_1\) plays the role of the unit operation \(\eta\), \(B_3\) plays the role of propagation through \(\kappa\), and \(B_2\) plays the role of contraction, since it combines two bounds into one.
Thus Oliva's formulation packages these requirements together at the level of bounded formulas, whereas our interface separates the formula-independent affine data \(\eta,\kappa\) from the formula-indexed contraction structure.

Another close relation is to the parametrized functional interpretation of Dinis and Oliva~\cite{DinisOliva2021Parametrised}.
It allows two type parameters, one for witnesses and one for challenge bounds, together with the corresponding compatibility relations and restriction predicates.
Our setting is more restrictive in two ways.
First, we use the same information type and compatibility relation for both witnesses and challenge bounds; this excludes interpretations such as modified realizability and the Diller--Nahm interpretation, which require these roles to be treated differently.
Second, we require the propagation data to be given by internal algebraic operations \(\eta\) and \(\kappa\); this excludes bounded functional interpretations, where the corresponding operations are defined externally.
The benefit of these restrictions is a smaller algebraic interface, well suited for propagating additional data through extracted realizers, as illustrated by the continuity example.

A further point of contact is the analysis of functional interpretations through linear logic~\cite{Oliva2008LinearDialectica,Oliva2010LinearIntuitionistic,FerreiraOliva2011ILL}.
There, structural behaviour is controlled by the exponential modality: the exponential-free fragment is treated uniformly, while different interpretations arise from different treatments of \(!A\), where weakening and contraction become available.
This is close in spirit to our separation between affine soundness and contraction.
The difference is that we work directly with affine finite-type arithmetic rather than through a linear-logic embedding: information nuclei account for affine propagation, while contraction structures account for the reuse of assumptions.

The information nuclei used here are also related to the parametrized translation of System~\(\T\)~\cite{Xu2020Gentzen}.
There, a nucleus controls a term translation: programs are translated so as to carry additional structure.
In the present paper the same idea is lifted from terms to formulas and proofs.
The nucleus now determines the information types, the compatibility relation used in bounded quantifiers, and the propagation operations used in the soundness proof.
Continuity nuclei, for instance, show that extracted realizers can carry continuity information in addition to their ordinary computational content.

\section{Conclusion}
\label{sec:conclusion}

We have shown that functional interpretations can be organized around two separate components: information nuclei for affine information propagation, and contraction structures for the reuse of assumptions.
The resulting framework is not intended to cover all known interpretations, but to provide a small interface in which these structural differences become visible.
It also suggests a route to functional interpretations in which extracted realizers are enriched with auxiliary data, by choosing information nuclei that propagate such data through the interpretation, as illustrated by the continuity example.

One natural continuation is to add extensionality to the account.
This paper deliberately separates it from affine information propagation and contraction: the standard Dialectica instance is recovered over the observational system \(\OHAo\), while weak and full extensionality are left untreated.
A satisfactory extension should identify what additional structure, beyond an information nucleus and a contraction structure, is needed to interpret extensionality rules.
This may require separating ordinary quantification from the computational information carried by realizers, along the lines of uniform Heyting arithmetic~\cite{Berger2004,RatiuSchwichtenberg2010Decorating} or nonstandard Heyting arithmetic~\cite{BBS2012}.

A second line of investigation is how information nuclei can be combined while preserving the compatibility relations and propagation laws needed for soundness.
Such combinations could provide a new way of combining multiple functional interpretations while also combining the auxiliary data carried by their extracted realizers, in the spirit of hybrid functional interpretations~\cite{HernestOliva2008Hybrid,Oliva2012HybridLinearIntuitionistic}.


\bibliography{refs}

\appendix

\section{Finite-type arithmetic systems}
\label{app:systems}

This appendix gives the formal details of the systems used in the paper.
All systems share the finite-type term language described in Section~\ref{sec:term-language}.
We first list the components from which the systems are built: the term equations, the affine intuitionistic rules, the contraction rule, the two treatments of equality, and the arithmetic principles.
The systems \(\NHAoa\), \(\NHAo\), and \(\OHAo\) are then defined by specifying which of these components they contain.

\paragraph*{Defining equations for terms}
The systems include the usual defining equations for the product projections,
Boolean case distinction, and primitive recursion as equational axiom schemas in
the object-language equality of the appropriate type.
We freely use these axioms as conversion rules in informal derivations:
\[
  \begin{array}{ll}
    \pr_1(\pair{a}{b}) =_\sigma a,
    \qquad &
    \pr_2(\pair{a}{b}) =_\tau b,
    \qquad\quad
    \pair{\pr_1(w)}{\pr_2(w)} =_{\sigma \times \tau} w,
    \\[.3em]
    \sfif_\sigma(\sft,a,b) =_\sigma a,
    \qquad &
    \sfif_\sigma(\sff,a,b) =_\sigma b,
    \\[.3em]
    \rec_\sigma(a,f,0) =_\sigma a,
    \qquad &
    \rec_\sigma(a,f,\suc n) =_\sigma f(n,\rec_\sigma(a,f,n)).
  \end{array}
\]
We freely use the \(\beta\)-conversion for lambda abstraction at the meta-level when writing terms.

\paragraph*{Sequents}
We formulate the systems in this appendix as single-conclusion sequent calculi.
Sequents have the form \(\Gamma \vdash A\), where \(\Gamma\) is a finite multiset of
formulas and \(A\) is a formula.  We write \(\Gamma,\Delta\) for multiset union,
and we allow the empty multiset as a context.  Each occurrence of a formula in
\(\Gamma\) is counted separately. Since contexts are finite multisets, exchange
is built into the representation of contexts rather than included as a rule.

\paragraph*{Affine intuitionistic rules}
Figure~\ref{fig:affine-rules} displays an affine version of the usual
single-conclusion intuitionistic rules.  The affine discipline concerns the use
of assumptions: an occurrence of an assumption may be used at most once.  Thus
the rules split contexts between premises whenever two premises are combined.
Weakening is allowed, but contraction is not.

\begin{figure}[htbp]
  \centering
  \(
  \begin{array}{c}
    \inferrule*[right=\(\mathrm{id}\)]
    { }
    {A \vdash A}
    \qquad
    \inferrule*[right=\(\mathrm{wkn}\)]
    {\Gamma \vdash B}
    {\Gamma,A \vdash B}
    \qquad
    \inferrule*[right=\(\mathrm{cut}\)]
    {\Gamma \vdash A \\ \Delta, A \vdash B}
    {\Gamma, \Delta \vdash B}
    \qquad
    \inferrule*[right=\(\mathrm{efq}\)]
    { }
    {\bot \vdash A}
    \\[.5em]

    \inferrule*[right=\(\wedge\mathrm{L}\)]
    {\Gamma,A,B \vdash C}
    {\Gamma,A\wedge B \vdash C}
    \qquad
    \inferrule*[right=\(\wedge\mathrm{R}\)]
    {\Gamma \vdash A \\ \Delta \vdash B}
    {\Gamma,\Delta \vdash A\wedge B}
    \\[.5em]

    \inferrule*[right=\(\vee\mathrm{L}\)]
    {\Gamma,A \vdash C \\ \Delta,B \vdash C}
    {\Gamma,\, \Delta,\, A \vee B \vdash C}
    \qquad
    \inferrule*[right=\(\vee\mathrm{R}_l\)]
    {\Gamma \vdash A}
    {\Gamma \vdash A\vee B}
    \qquad
    \inferrule*[right=\(\vee\mathrm{R}_r\)]
    {\Gamma \vdash B}
    {\Gamma \vdash A\vee B}
    \\[.5em]

    \inferrule*[right=\(\mathord{\to}\mathrm{L}\)]
    {\Gamma \vdash A \\ \Delta, B \vdash C}
    {\Gamma,\, \Delta,\, A \to B \vdash C}
    \qquad
    \inferrule*[right=\(\mathord{\to}\mathrm{R}\)]
    {\Gamma, A \vdash B}
    {\Gamma \vdash A \to B}
    \\[.5em]

    \inferrule*[right=\(\forall\mathrm{L}\)]
    {\Gamma, A(t) \vdash B}
    {\Gamma,\, \forall x^\sigma A(x) \vdash B}
    \qquad
    \inferrule*[right=\(\forall\mathrm{R}\)\quad\(x \notin \FV(\Gamma)\)]
    {\Gamma \vdash A(x)}
    {\Gamma \vdash \forall x^\sigma A(x)}
    \\[.5em]

    \inferrule*[right=\(\exists\mathrm{L}\)\quad\(x \notin \FV(\Gamma\text{,}\,B)\)]
    {\Gamma, A(x) \vdash B}
    {\Gamma, \exists x^\sigma A(x) \vdash B}
    \qquad
    \inferrule*[right=\(\exists\mathrm{R}\)]
    {\Gamma \vdash A(t)}
    {\Gamma \vdash \exists x^\sigma A(x)}
  \end{array}
  \)
  \caption{Affine structural and intuitionistic logical rules.}
  \label{fig:affine-rules}
\end{figure}

\paragraph*{Contraction}
The contraction rule is
\[
  \inferrule*[right=\(\mathrm{contr}\)]
  {\Gamma,A,A \vdash B}
  {\Gamma,A \vdash B}.
\]

\paragraph*{Neutral equality}
For the neutral systems \(\NHAoa\) and \(\NHAo\), equality \(=_\sigma\) is primitive at every finite type \(\sigma\), and is governed by the equivalence-relation axioms
\[
  s =_\sigma s,
  \qquad
  s =_\sigma t \to t =_\sigma s,
  \qquad
  s =_\sigma t \wedge t =_\sigma u \to s =_\sigma u.
\]
In addition, equality is preserved by term contexts: for every term \(r:\tau\)
and variable \(x:\sigma\),
\[
  s =_\sigma t \to r[x \eqsubst s] =_\tau r[x \eqsubst t].
\]

\paragraph*{Observational equality}
In \(\OHAo\), only equality at the base types \(\N\), \(\Two\), and \(\One\) is primitive.
Equality at product and function types is defined recursively by
\[
  x =_{\sigma\times\tau} y \eqdef
  \pr_1(x) =_\sigma \pr_1(y) \wedge \pr_2(x) =_\tau \pr_2(y),
  \qquad
  f =_{\sigma\to\tau} g \eqdef
  \forall x^\sigma\, f(x) =_\tau g(x).
\]
The equivalence-relation laws for these observational equality predicates are proved by induction on types from the corresponding laws for the primitive base-type equalities.
No additional congruence or extensionality rule for arbitrary term contexts is included.

\paragraph*{Arithmetic principles}
The arithmetic principles are stated using the equality predicates of the system under consideration: neutral equality for \(\NHAoa\) and \(\NHAo\), and observational equality for \(\OHAo\).
They consist of the arithmetic axioms together with the closed induction rule.
The arithmetic axioms comprise the defining equations for term constants listed above and the separation and injectivity axioms for the constructors of the base types,
\[
  \sft =_\Two \sff \to \bot,
  \qquad
  \suc x =_\N 0 \to \bot,
  \qquad
  \suc x =_\N \suc y \to x =_\N y.
\]
We assume the usual decidability of equality at the base types.

Induction is included, for all formulas, as the closed rule
\[
  \inferrule*[right=\(\mathrm{ind}\)]
  {\vdash A(0) \\
   \vdash \forall n^\N(A(n)\to A(\suc n))}
  {\vdash \forall n^\N A(n)}.
\]
The closed formulation ensures that induction itself does not introduce reuse of assumptions into the affine system.  In systems with contraction, the usual induction axiom, equivalently the induction rule with arbitrary context, is derivable from this closed rule.

\paragraph*{The systems}
The three intuitionistic arithmetic systems in all finite types used in the paper are obtained by assembling the preceding components as follows:
\begin{itemize}
  \item \(\NHAoa\) consists of the affine intuitionistic rules of Figure~\ref{fig:affine-rules}, neutral equality, and the arithmetic principles listed above.
  \item \(\NHAo\) is obtained from \(\NHAoa\) by adding the contraction rule.
  \item \(\OHAo\) consists of the affine intuitionistic rules of Figure~\ref{fig:affine-rules} together with contraction, observational equality, and the arithmetic principles listed above, with equality in those principles read observationally.
\end{itemize}

\section{Soundness proof details}
\label{app:proofs}

This appendix gives the proof details for the affine soundness theorem
(Theorem~\ref{thm:affine-soundness}) and the full soundness theorem
(Theorem~\ref{thm:full-soundness}).  The affine soundness theorem is proved by
induction on \(\NHAoa\)-derivations.  We group the cases into affine structural
rules, affine logical rules, equality axioms and defining equations, and the
closed induction rule.  The full soundness theorem is then obtained by adding the
single contraction case.

\paragraph*{Auxiliary lemmas}
We first record two elementary lemmas that will be used repeatedly in the soundness proof.
The first lemma provides canonical information objects and realizers.

\begin{lemma}[Canonical information and realizers]
  \label{lem:canonical}
  For every finite type \(\sigma\), there is a closed term
  \(\cc_\sigma:\sigma^\nucleus\);
  consequently, for every formula \(A\), there is a closed term
  \(\cc_A:\ty_A^+\).
  In particular, the canonical information objects are inhabited.
\end{lemma}

\begin{proof}
  First choose, by recursion on finite types, a closed term
  \(d_\sigma:\sigma\) of each finite type: for instance
  \(d_\N\eqdef 0\), \(d_\Two\eqdef\sft\), \(d_\One\eqdef\star\),
  \(d_{\sigma\times\tau}\eqdef\pair{d_\sigma}{d_\tau}\), and
  \(d_{\sigma\to\tau}\eqdef\lambda x^\sigma.\,d_\tau\).  Define
  \(\cc_\sigma\eqdef\eta_\sigma(d_\sigma)\).  Then
  \(d_\sigma\compat\cc_\sigma\) follows from (IN1), hence
  \(\cc_\sigma\) is inhabited.

  The canonical realizer \(\cc_A:\ty_A^+\) is defined by recursion on \(A\):
  \[
    \begin{array}{rcl}
      \cc_P                    & \eqdef & \star,                               \\[.2em]
      \cc_{B\wedge C}          & \eqdef & \pair{\cc_B}{\cc_C},                 \\[.2em]
      \cc_{B\vee C}            & \eqdef & \langle \cc_\Two,\cc_B,\cc_C\rangle, \\[.2em]
      \cc_{B\to C}             & \eqdef &
      \left\langle
      \lambda x^{\ty_B^+}.\,\cc_C,
      \lambda w^{\ty_B^+\times\ty_C^-}.\,\cc_{\ty_B^-}
      \right\rangle,                                                           \\[.2em]
      \cc_{\forall x^\sigma B} & \eqdef &
      \lambda a^{\sigma^\nucleus}.\,\cc_B,                                     \\[.2em]
      \cc_{\exists x^\sigma B} & \eqdef &
      \pair{\cc_\sigma}{\cc_B}.
    \end{array}
  \]
\end{proof}

The second lemma uniformizes challenge bounds: if the bounds for a context depend on a challenge variable, then \(\kappa\) turns them into bounds that are valid uniformly for all challenges covered by a given information object.

\begin{lemma}[Challenge bound uniformization]
  \label{lem:uniform-bound}
  Suppose that, for each \(a:\rho\), we have bounds
  \(u_i(a):(\ty^-_{C_i})^\nucleus\) such that
  \[
    \forall \vec z\compat \vec u(a)\,|\Gamma|^{\vec x}_{\vec z}
    \vdash
    D(a).
  \]
  Given \(w:\rho^\nucleus\), define for each \(C_i \in \Gamma\)
  \[
    U_i \eqdef (\lambda a^\rho.\,u_i(a))^\kappa(w) : (\ty^-_{C_i})^\nucleus.
  \]
  Then
  \[
    \forall \vec z\compat \vec U\,|\Gamma|^{\vec x}_{\vec z}
    \vdash
    \forall a\compat w\,D(a).
  \]
  Moreover, if \(w\) is inhabited and each \(u_i(a)\) is inhabited whenever \(a\compat w\), then each \(U_i\) is inhabited.
\end{lemma}
\begin{proof}
  Assume \(\forall \vec z\compat \vec U\,|\Gamma|^{\vec x}_{\vec z}\), and let
  \(a\compat w\).  For each \(i\), \textup{(IN2)} applied to
  \(\lambda a.\,u_i(a)\) gives
  \[
    u_i(a)\nule U_i .
  \]
  Hence bounded monotonicity turns the assumptions bounded by \(U_i\) into the
  assumptions bounded by \(u_i(a)\).  The displayed premise then yields
  \(D(a)\).  Since \(a\compat w\) was arbitrary, we obtain
  \(\forall a\compat w\,D(a)\).

  For inhabitedness, choose \(a\) with \(a\compat w\).  By assumption,
  \(u_i(a)\) is inhabited.  Since \(u_i(a) \nule U_i\), inhabitedness transfers to
  \(U_i\).
\end{proof}

We will need a derived pairing operation on information objects, for example in the proof of the conjunction-left case.

\begin{definition}[Information pairing operation]
For \(a:\sigma^\nucleus\) and \(b:\tau^\nucleus\), define
\[
  a\otimes b
  \eqdef
  \bigl(\lambda x^\sigma.\,
    (\lambda y^\tau.\,\eta_{\sigma\times\tau}(\pair{x}{y}))^\kappa(b)
  \bigr)^\kappa(a)
  :(\sigma\times\tau)^\nucleus .
\]
\end{definition}

\begin{lemma}[Pairing information]
  \label{lem:pairing-information}
  If \(x\compat a\) and \(y\compat b\), then
  \[
    \pair{x}{y}\compat a\otimes b .
  \]
  Consequently, if \(a\) and \(b\) are inhabited, then \(a\otimes b\) is inhabited.
\end{lemma}

\begin{proof}
  By (IN1), \(\pair{x}{y}\compat\eta(\pair{x}{y})\).  From
  \(y\compat b\), (IN2) applied to
  \(\lambda y.\eta(\pair{x}{y})\) gives
  \[
    \eta(\pair{x}{y})
    \nule
    (\lambda y.\eta(\pair{x}{y}))^\kappa(b).
  \]
  From \(x\compat a\), (IN2) applied to
  \(\lambda x.(\lambda y.\eta(\pair{x}{y}))^\kappa(b)\) gives
  \[
    (\lambda y.\eta(\pair{x}{y}))^\kappa(b)
    \nule
    a\otimes b .
  \]
  Hence \(\pair{x}{y}\compat a\otimes b\) by transitivity of \(\nule\).
  The inhabitedness statement follows by taking compatible witnesses for \(a\)
  and \(b\).
\end{proof}

\paragraph*{Affine structural rules}
We consider three affine structural rules listed in Figure~\ref{fig:affine-rules}.

\medskip
\noindent
\textbf{Identity.}
The identity rule is realized by the realizer variable for the unique assumption.
For \(A\vdash A\), let \(x:\ty_A^+\) be the realizer variable for the assumption \(A\) and let \(y:\ty_A^-\) be the challenge variable for the conclusion \(A\). The realizer for the conclusion is \(x\), and the challenge bound for the assumption is \(\eta(y)\). By (IN1), the bound \(\eta(y)\) is inhabited and the translated assumption \(\forall z \compat \eta(y) |A|^x_z \) gives \(|A|^x_y\).

\medskip
\noindent
\textbf{Weakening.}
Suppose the last rule is weakening
\[
  \inferrule*[right=\(\mathrm{wkn}\)]
  {\Gamma\vdash B}
  {\Gamma,A\vdash B}.
\]
If the induction hypothesis realizes \(\Gamma\vdash B\), then for the weakened sequent \(\Gamma,A\vdash B\) we keep the same realizer for \(B\) and the same challenge bounds for the assumptions in \(\Gamma\).
The unused assumption \(A\) is assigned the canonical inhabited challenge bound \(\cc\) from Lemma~\ref{lem:canonical}.

\medskip
\noindent
\textbf{Cut.}
Suppose the last rule is cut
\[
  \inferrule*[right=\(\mathrm{cut}\)]
  {\Gamma\vdash A \\ \Delta,A\vdash B}
  {\Gamma,\Delta\vdash B}.
\]
The induction hypothesis for \(\Gamma \vdash A\) gives a realizer \(r:\ty_A^+\) for \(A\) and challenge bounds \(u_i:(\ty^-_{C_i})^\nucleus\) for the assumptions \(C_i \in \Gamma\), where \(r\) depends only on the realizer variables \(x_i : \ty_{C_i}^+\) for the assumptions \(\Gamma\) and \(u_i\) depends on the same realizer variables as well as the challenge variable \(a:\ty_A^-\) to \(A\), such that
\[
  \forall \vec z\compat \vec u\,|\Gamma|^{\vec x}_{\vec z}
  \vdash |A|^r_a .
\]
The induction hypothesis for \(\Delta,A\vdash B\) gives a realizer \(s:\ty_B^+\) for \(B\), and challenge bounds \(v_j:(\ty^-_{D_j})^\nucleus\) for the assumptions \(D_j \in \Delta\) and \(w:(\ty^-_A)^\nucleus\) for the distinguished assumption~\(A\), where \(s\) depends only on the realizer variables \(y_j : \ty_{D_j}^+\) for \(\Delta\) and \(e:\ty^+_A\) for \(A\), and \(v_j\) and \(w\) depend on the same realizer variables as well as the challenge variable \(b:\ty_B^-\) to \(B\), such that
\[
  \forall \vec{z} \compat \vec{v} \, |\Delta|^{\vec y}_{\vec z} ,\,
  \forall z \compat w \, |A|^{e}_z
  \vdash |B|^s_b .
\]
Note that \(A\) is eliminated in the conclusion sequent and thus the extracted realizer and challenge bounds cannot depend on the variables \(e\) or \(a\) for \(A\). We define
\[
  t \eqdef s[e\eqsubst r] : \ty_B^+,
  \qquad
  w' \eqdef w[e\eqsubst r] : (\ty^-_A)^\nucleus,
  \qquad
  V_j \eqdef v_j[e\eqsubst r] : (\ty^-_{D_j})^\nucleus .
\]
In particular, we define the challenge bounds for the assumptions in \(\Gamma\) as follows:
\[
  U_i \eqdef (\lambda a.\,u_i)^\kappa(w') : (\ty^-_{C_i})^\nucleus .
\]
Applying Lemma~\ref{lem:uniform-bound} componentwise to the first induction
hypothesis, with the bounds \(U_i\) defined above, gives
\[
  \forall \vec z\compat \vec U\,|\Gamma|^{\vec x}_{\vec z}
  \vdash
  \forall z\compat w'\, |A|^r_z .
\]
Substituting \(e\eqsubst r\) in the second induction hypothesis gives
\[
  \forall \vec z \compat \vec V\, |\Delta|^{\vec y}_{\vec z},\,
  \forall z \compat w'\, |A|^r_z
  \vdash
  |B|^t_b .
\]
Cutting these two sequents yields
\[
  \forall \vec z\compat \vec U\, |\Gamma|^{\vec x}_{\vec z},\,
  \forall \vec z\compat \vec V\, |\Delta|^{\vec y}_{\vec z}
  \vdash
  |B|^t_b
\]
which is the desired translation of \(\Gamma,\Delta\vdash B\), with the realizer \(t\) and the challenge bounds \(U_i\) and \(V_j\) for the assumptions in \(\Gamma\) and \(\Delta\), respectively.

\paragraph*{Affine logical rules}
We now consider the logical rules of Figure~\ref{fig:affine-rules}.

\medskip
\noindent
\textbf{Ex falso.}
The case of ex falso is immediate.  The rule is
\[
  \inferrule*[right=\(\mathrm{efq}\)]
  { }
  {\bot \vdash A}.
\]
Note that \(\bot\) has trivial realizer and challenge types. Let \(x: \One\) be the realizer variable for the assumption \(\bot\), and let \(y:\ty_A^-\) be the challenge to the conclusion \(A\).  We take the conclusion realizer to be the canonical realizer \(\cc_A:\ty_A^+\) from Lemma~\ref{lem:canonical}.  The challenge bound for the assumption \(\bot\) is \(\eta_\One(\star):\One^\nucleus\), which is inhabited by (IN1).  From the translated assumption
\[
  \forall z\compat \eta_\One(\star)\,|\bot|^x_z
\]
we obtain \(\bot\), since \(\star\compat\eta_\One(\star)\) by (IN1) and \(|\bot|^x_\star\) is just \(\bot\).  Hence \(|A|^{\cc_A}_y\) follows by ex falso in the target system.

\medskip
\noindent
\textbf{Conjunction left.}
The case of conjunction left uses the information pairing operation.  The rule is
\[
  \inferrule*[right=\(\wedge\mathrm{L}\)]
  {\Gamma,A,B \vdash C}
  {\Gamma,A\wedge B \vdash C}.
\]
Let \(x:\ty_{A\wedge B}^+\) be the realizer variable for the assumption
\(A\wedge B\), and write \(x\equiv\langle x_1,x_2\rangle\), where
\(x_1:\ty_A^+\) and \(x_2:\ty_B^+\).  By the induction hypothesis for the
premise, with realizer variables \(x_1\) and \(x_2\) for the assumptions \(A\)
and \(B\), we obtain a realizer \(r:\ty_C^+\) for \(C\), challenge bounds
\(\vec u\) for the assumptions in \(\Gamma\), and challenge bounds
\(v:(\ty_A^-)^\nucleus\) and \(w:(\ty_B^-)^\nucleus\) for the assumptions \(A\)
and \(B\), such that
\[
  \forall \vec z\compat\vec u\,|\Gamma|^{\vec a}_{\vec z}, \,
  \forall z\compat v\,|A|^{x_1}_z, \,
  \forall z\compat w\,|B|^{x_2}_z
  \vdash
  |C|^r_y .
\]
For the conclusion sequent \(\Gamma,A\wedge B \vdash C\), we use the same realizer \(r\) for \(C\) and the
same bounds \(\vec u\) for \(\Gamma\).  The challenge bound assigned to the
single assumption \(A\wedge B\) is the paired information object
\[
  v\otimes w : (\ty_A^-\times\ty_B^-)^\nucleus .
\]
It is inhabited by Lemma~\ref{lem:pairing-information}, because \(v\) and \(w\) are inhabited by the induction hypothesis.
Assume now the translated assumption
\[
  H:\forall q\compat v\otimes w\, |A\wedge B|^x_q .
\]
We derive the two assumptions needed for the induction hypothesis.  To prove
\(\forall z\compat v\,|A|^{x_1}_z\), let \(z\compat v\).  Since \(w\) is
inhabited, choose \(z_B:\ty_B^-\) with \(z_B\compat w\).  By
Lemma~\ref{lem:pairing-information},
\[
  \pair{z}{z_B}\compat v\otimes w .
\]
Hence \(H\) gives \(|A\wedge B|^x_{\pair{z}{z_B}}\),
unfolding to \(|A|^{x_1}_z \wedge |B|^{x_2}_{z_B}\); thus, \(|A|^{x_1}_z\).
Since \(z\compat v\) was arbitrary, we obtain \(\forall z\compat v\,|A|^{x_1}_z\).
The derivation of \(\forall z\compat w\,|B|^{x_2}_z\) is symmetric, using the inhabitedness of \(v\).
Applying the induction hypothesis for the premise then yields \(|C|^r_y\), as required.

\medskip
\noindent
\textbf{Conjunction right.}
The case of conjunction right is componentwise.  The rule is
\[
  \inferrule*[right=\(\wedge\mathrm{R}\)]
  {\Gamma \vdash A \\
   \Delta \vdash B}
  {\Gamma,\Delta \vdash A\wedge B}.
\]
Let \(y:\ty^-_{A\wedge B}\) be a challenge variable to the conclusion \(A \wedge B\), and write \(y\equiv\pair{y_1}{y_2}\) with
\(y_1:\ty_A^-\) and \(y_2:\ty_B^-\).  By the induction hypothesis for the first
premise, applied with challenge \(y_1\), we obtain a realizer \(r:\ty_A^+\) and
challenge bounds \(\vec u\) for the assumptions in \(\Gamma\), such that
\[
  \forall \vec z\compat\vec u\,|\Gamma|^{\vec x}_{\vec z}
  \vdash
  |A|^r_{y_1}.
\]
By the induction hypothesis for the second premise, applied with challenge
\(y_2\), we obtain a realizer \(s:\ty_B^+\) and challenge bounds \(\vec v\) for
the assumptions in \(\Delta\), such that
\[
  \forall \vec z\compat\vec v\,|\Delta|^{\vec w}_{\vec z}
  \vdash
  |B|^s_{y_2}.
\]
For the conclusion sequent, the realizer is
\[
  \pair{r}{s}:\ty^+_{A\wedge B}.
\]
The challenge bounds for the assumptions in \(\Gamma\) are \(\vec u\), and those
for the assumptions in \(\Delta\) are \(\vec v\).
They are inhabited by the induction hypotheses.
From the translated assumptions for \(\Gamma,\Delta\), the two induction hypotheses give
\(|A|^r_{y_1}\) and \(|B|^s_{y_2}\); thus, \(|A\wedge B|^{\pair{r}{s}}_y\) by unfolding the translation of conjunction.

\medskip

\medskip
\noindent
\textbf{Disjunction left.}
The case of disjunction left uses the lifted case distinction.  The rule is
\[
  \inferrule*[right=\(\vee\mathrm{L}\)]
  {\Gamma,A \vdash C \\
   \Delta,B \vdash C}
  {\Gamma,\Delta,A\vee B \vdash C}.
\]
Let \(x:\ty^+_{A\vee B}\) be the realizer variable for the assumption
\(A\vee B\), and write \(x\equiv\langle x_1,x_2,x_3\rangle\),
where \(x_1:\Two^\nucleus\), \(x_2:\ty_A^+\), and \(x_3:\ty_B^+\).
By the induction hypothesis for the left premise, with \(x_2\) as the realizer
variable for the assumption \(A\), we obtain a realizer \(r:\ty_C^+\), challenge
bounds \(\vec u\) for the assumptions in \(\Gamma\), and a challenge bound
\(v:(\ty_A^-)^\nucleus\) for \(A\), such that
\[
  \forall \vec z\compat\vec u\,|\Gamma|^{\vec a}_{\vec z},
  \forall z\compat v\,|A|^{x_2}_z
  \vdash
  |C|^r_y .
\]
Similarly, the induction hypothesis for the right premise, with \(x_3\) as the
realizer variable for the assumption \(B\), gives a realizer \(s:\ty_C^+\),
challenge bounds \(\vec w\) for the assumptions in \(\Delta\), and a challenge
bound \(v':(\ty_B^-)^\nucleus\) for \(B\), such that
\[
  \forall \vec z\compat\vec w\,|\Delta|^{\vec b}_{\vec z},
  \forall z\compat v'\,|B|^{x_3}_z
  \vdash
  |C|^s_y .
\]
For the conclusion sequent, the realizer for \(C\) is
\[
  q \eqdef \lif_C(x_1,r,s):\ty_C^+ .
\]
The challenge bounds for \(\Gamma\) and \(\Delta\) are respectively \(\vec u\)
and \(\vec w\).  The inhabited challenge bound for the assumption \(A\vee B\) is given by
\[
  \eta\bigl(\pair{v}{v'}\bigr)
  : \bigl((\ty_A^-)^\nucleus\times(\ty_B^-)^\nucleus\bigr)^\nucleus .
\]
Assume the translated assumptions for \(\Gamma\) and \(\Delta\), and assume also
\[
  \forall p\compat\eta(\pair{v}{v'})\, |A\vee B|^x_p.
\]
By (IN1), \(\pair{v}{v'}\compat\eta(\pair{v}{v'})\), so the above assumption gives
\(|A\vee B|^x_{\pair{v}{v'}}\), unfolding to
\[
  (\sft \compat x_1 \wedge \forall z\compat v\,|A|^{x_2}_z) \vee
  (\sff \compat x_1 \wedge \forall z\compat v'\,|B|^{x_3}_z).
\]
In the left case, together with the translated assumptions for \(\Gamma\), the left induction
hypothesis gives \(|C|^r_y\).  Since \(\sft\compat x_1\), correctness of lifted case distinction (Corollary~\ref{cor:lifted-if}) gives \(r\nule_C q\), and monotonicity in realizers (Lemma~\ref{lem:formula-monotonicity}) yields \(|C|^q_y\).
Similarly, the right case also gives \(|C|^q_y\); thus, \(|C|^q_y\) follows in both cases.

\medskip
\noindent
\textbf{Disjunction right.}
The case of disjunction right is similar for the two injections; we spell
out the left injection.  The rule is
\[
  \inferrule*[right=\(\vee\mathrm{R}_l\)]
  {\Gamma \vdash A}
  {\Gamma \vdash A\vee B}.
\]
Let \(y: (\ty^-_A)^\nucleus \times (\ty^-_B)^\nucleus\) be the challenge variable to the conclusion \(A\vee B\), and write \(y\equiv \pair{y_1}{y_2}\),
where \(y_1:(\ty_A^-)^\nucleus\) and \(y_2:(\ty_B^-)^\nucleus\).
By the induction hypothesis for the premise, we obtain a realizer \(r:\ty_A^+\) for \(A\) and challenge bounds \(\vec u\) for the assumptions in \(\Gamma\) where \(\vec u\) may depend on the challenge variable \(z:\ty^-_A\), such that
\[
  \forall \vec z\compat \vec u\,|\Gamma|^{\vec x}_{\vec z}
  \vdash
  |A|^r_z .
\]
For the conclusion sequent, we define the realizer
\[
  q \eqdef \langle \eta(\sft), r, \cc_B\rangle : \ty^+_{A\vee B}
\]
where \(\cc_B:\ty_B^+\) is the canonically chosen realizer from
Lemma~\ref{lem:canonical}, and define the challenge bounds for the assumptions in \(\Gamma\) by
\[
  U_i \eqdef (\lambda z.\,u_i)^\kappa(y_1) : (\ty_{C_i}^-)^\nucleus
  \qquad\text{for each } C_i \in \Gamma.
\]
By Lemma~\ref{lem:uniform-bound}, these bounds are inhabited, assuming \(y_1\) is inhabited, and
\[
  \forall \vec z\compat \vec U\,|\Gamma|^{\vec x}_{\vec z}
  \vdash
  \forall z \compat y_1 \, |A|^r_z .
\]
Because \(\sft \compat \eta(\sft)\) by (IN1), the conclusion of the above sequent becomes
\[
  \sft \compat \eta(\sft) \wedge \forall z \compat y_1 \, |A|^r_z
\]
which is exactly the left disjunct of \(|A\vee B|^q_y\) after unfolding the translation of disjunction.
Hence
\[
  \forall \vec z\compat \vec U\,|\Gamma|^{\vec x}_{\vec z}
  \vdash
  |A\vee B|^q_y
\]
follows, as required.
The right injection \(\vee\mathrm{R}_r\) is symmetric.

\medskip
\noindent
\textbf{Implication left.}
The case of implication left requires to propagate challenge bounds.  The rule is
\[
  \inferrule*[right=\(\mathord{\to}\mathrm{L}\)]
  {\Gamma \vdash A \\
   \Delta,B \vdash C}
  {\Gamma,\Delta,A\to B \vdash C}.
\]
Let \(x:\ty^+_{A\to B}\) be the realizer variable for the assumption
\(A\to B\), and write \( x\equiv\langle x_1,x_2\rangle\),
where \(x_1:\ty_A^+\to\ty_B^+\), and \(x_2:\ty_A^+\times\ty_B^-\to(\ty_A^-)^\nucleus\).
Let \(y:\ty_C^-\) be the challenge to the conclusion \(C\).

By the induction hypothesis for the first premise \(\Gamma\vdash A\), with
challenge variable \(a:\ty_A^-\), we obtain a realizer \(r:\ty_A^+\) and
challenge bounds \(\vec u\) for the assumptions in \(\Gamma\), such that
\[
  \forall \vec z\compat\vec u\,|\Gamma|^{\vec g}_{\vec z}
  \vdash
  |A|^r_a .
\]
Here the bounds \(u_i\) may depend on \(a\), while \(r\) does not.  By the
induction hypothesis for the second premise \(\Delta,B\vdash C\), with realizer
variable \(e:\ty_B^+\) for the assumption \(B\), we obtain a realizer
\(s:\ty_C^+\), challenge bounds \(\vec v\) for the assumptions in \(\Delta\), and a challenge bound
\(w:(\ty_B^-)^\nucleus\) for \(B\), such that
\[
  \forall \vec z\compat\vec v\,|\Delta|^{\vec d}_{\vec z},
  \forall z\compat w\,|B|^e_z
  \vdash
  |C|^s_y .
\]
We substitute \(e\eqsubst x_1(r)\) in the second induction hypothesis and write
\[
  s'\eqdef s[e\eqsubst x_1(r)],\qquad
  v'_j\eqdef v_j[e\eqsubst x_1(r)],\qquad
  w'\eqdef w[e\eqsubst x_1(r)].
\]
Thus it remains to derive
\(\forall z\compat w'\,|B|^{x_1(r)}_z\).

For the assumption \(A\to B\) in the conclusion sequent, we use the challenge bound
\[
  P \eqdef (\lambda z^{\ty_B^-}.\,\eta(\pair{r}{z}))^\kappa(w')
  :(\ty_A^+\times\ty_B^-)^\nucleus .
\]
For the assumptions in \(\Gamma\), we first collect all challenges to \(A\) that
may be requested by the implication assumption against challenges to \(B\) below
\(w'\):
\[
  W_A \eqdef (\lambda z^{\ty_B^-}.\,x_2(r,z))^\kappa(w')
  : (\ty_A^-)^\nucleus .
\]
Then define the challenge bounds for \(G_i\in\Gamma\) by
\[
  U_i \eqdef (\lambda a^{\ty_A^-}.\,u_i)^\kappa(W_A)
  : (\ty_{G_i}^-)^\nucleus .
\]
By Lemma~\ref{lem:uniform-bound} applied to the first induction hypothesis, the translated assumptions for \(\Gamma\), bounded by \(\vec U\), yield
\[
  \forall a\compat W_A\,|A|^r_a .
\]

Assume now the translated assumptions for \(\Gamma\), bounded by \(\vec U\), the
translated assumptions for \(\Delta\), bounded by \(\vec v'\), and the translated
assumption for \(A\to B\), bounded by \(P\):
\[
  H:\forall p\compat P\,|A\to B|^x_p .
\]
We show \(\forall z\compat w'\,|B|^{x_1(r)}_z\).  Let \(z\compat w'\).  By (IN1), we have \(\pair{r}{z}\compat \eta(\pair{r}{z})\), and by (IN2) applied to \(\lambda z.\eta(\pair{r}{z})\), we have
\[
  \pair{r}{z}\compat P .
\]
Hence \(H\) yields \(|A\to B|^x_{\pair{r}{z}}\), which unfolds to
\[
  \bigl(\forall a\compat x_2(r,z)\,|A|^r_a\bigr)
  \to
  |B|^{x_1(r)}_z .
\]
It remains to prove the antecedent.  If \(a\compat x_2(r,z)\), then
(IN2) applied to \(\lambda z.\,x_2(r,z)\), together with
\(z\compat w'\), gives
\[
  x_2(r,z)\nule W_A .
\]
Therefore \(a\compat W_A\), and the uniformized conclusion
\(\forall a\compat W_A\,|A|^r_a\) gives \(|A|^r_a\).  Thus
\(|B|^{x_1(r)}_z\).  Since \(z\compat w'\) was arbitrary, we have
\[
  \forall z\compat w'\,|B|^{x_1(r)}_z .
\]
The second induction hypothesis substituted by \(e:=x_1(r)\) now yields \(|C|^{s'}_y\).

Thus the conclusion sequent is realized by \(s'\), with challenge bounds
\(\vec U\) for \(\Gamma\), \(\vec v'\) for \(\Delta\), and \(P\) for the
assumption \(A\to B\).  Inhabitedness of these bounds follows from the induction
hypotheses and axioms (IN1) and (IN2).

\medskip
\noindent
\textbf{Implication right.}
The case of implication right packages the realizer and challenge-bound
transformers obtained from the premise.  The rule is
\[
  \inferrule*[right=\(\mathord{\to}\mathrm{R}\)]
  {\Gamma,A \vdash B}
  {\Gamma \vdash A\to B}.
\]
Let \(y:\ty^-_{A\to B}\) be a challenge variable to the conclusion \(A \to B\), and write
\(y\equiv\pair{a}{b}\), where \(a:\ty_A^+\) and \(b:\ty_B^-\).
By the induction hypothesis for the premise \(\Gamma,A \vdash B\), with \(a\) as the realizer variable for the additional assumption \(A\) and \(b\) as the challenge variable for the conclusion \(B\), we obtain a realizer \(r:\ty_B^+\), challenge bounds
\(u_i:(\ty_{G_i}^-)^\nucleus\) for the assumptions \(G_i\in\Gamma\), and a
challenge bound \(v:(\ty_A^-)^\nucleus\) for \(A\), such that
\[
  \forall \vec z\compat\vec u\,|\Gamma|^{\vec g}_{\vec z},
  \forall z\compat v\,|A|^a_z
  \vdash
  |B|^r_b .
\]
Here \(r\) may depend on the realizer variables for \(\Gamma\) and on \(a\), but
not on \(b\), by the variable condition of the induction hypothesis.  The bounds
\(\vec u\) and \(v\) may depend on both \(a\) and \(b\).
For the conclusion \(A\to B\), define
\[
  h\eqdef
  \left\langle
    \lambda a^{\ty_A^+}.\,r,\;
    \lambda y^{\ty_A^+\times\ty_B^-}.\,v
  \right\rangle
  :\ty^+_{A\to B},
\]
where in the first component \(r\) is the premise realizer obtained with
realizer variable \(a\), and in the second component \(v\) is the corresponding
challenge bound for the assumption \(A\), with \(y=\pair{a}{b}\).
For each assumption \(G_i\in\Gamma\) in the conclusion sequent, the challenge bound is
\[
  U_i\eqdef u_i :(\ty_{G_i}^-)^\nucleus,
\]
again with \(a\) and \(b\) read as the components of the conclusion challenge \(y\).
These bounds are inhabited by the induction hypothesis for the premise.
Assume the translated assumptions for \(\Gamma\), bounded by \(\vec U\).
To prove the translated conclusion, unfold the implication clause:
\[
  |A\to B|^h_y
  \equiv
  \bigl(\forall z\compat h_2(y)\,|A|^{a}_z\bigr)
  \to
  |B|^{h_1(a)}_b .
\]
Since \(h_2(y)=v\) and \(h_1(a)=r\), this is exactly
\[
  \bigl(\forall z\compat v\,|A|^{a}_z\bigr)
  \to
  |B|^r_b .
\]
The latter follows from the induction hypothesis for the premise, together with
the translated assumptions for \(\Gamma\).  Hence \(|A\to B|^h_y\).

\medskip
\noindent
\textbf{Universal left.}
The case of universal left uses the exact information attached to
the chosen instance.  The rule is
\[
  \inferrule*[right=\(\forall\mathrm{L}\)]
  {\Gamma,A(t)\vdash B}
  {\Gamma,\forall x^\sigma A(x)\vdash B}.
\]
Let \(p: \sigma^\nucleus\to\ty_A^+\) be a realizer variable for the assumption \(\forall x^\sigma A(x)\). By the induction hypothesis for the premise, with \(p(\eta(t))\) as the
realizer for the assumption \(A(t)\), we obtain a realizer
\(r:\ty_B^+\), challenge bounds \(\vec u\) for the assumptions in \(\Gamma\),
and a challenge bound \(v:(\ty_A^-)^\nucleus\) for \(A(t)\), such that
\[
  \forall \vec z\compat\vec u\,|\Gamma|^{\vec g}_{\vec z},
  \forall z\compat v\, |A(t)|^{p(\eta(t))}_z
  \vdash
  |B|^r_y .
\]
For the conclusion sequent, we use the same realizer \(r\) for \(B\), the same
bounds \(\vec u\) for \(\Gamma\), and the following challenge bound for the
assumption \(\forall x^\sigma A(x)\):
\[
  W \eqdef (\lambda z^{\ty_A^-}.\,\eta(\pair{\eta(t)}{z}))^\kappa(v)
  : (\sigma^\nucleus\times\ty_A^-)^\nucleus,
\]
which is inhabited by (IN1), (IN2) and the induction hypothesis for the premise.
Assume the translated assumption
\[
  H:\forall q\compat W\,|\forall x^\sigma A(x)|^p_q .
\]
We derive the assumption needed for the premise induction hypothesis.
Let \(z\compat v\).
We have
\[
\pair{\eta(t)}{z}\compat\eta(\pair{\eta(t)}{z})
\]
by (IN1), and
\[
\eta(\pair{\eta(t)}{z})\nule W
\]
by (IN2), hence
\[
\pair{\eta(t)}{z}\compat W.
\]
Then \(H\) gives \(|\forall x^\sigma A(x)|^p_{\pair{\eta(t)}{z}}\),
unfolding to
\[
  \forall x\compat\eta(t)\, |A(x)|^{p(\eta(t))}_z .
\]
Since \(t\compat\eta(t)\) by (IN1), we obtain
\[
  |A(t)|^{p(\eta(t))}_z .
\]
Because \(z\compat v\) was arbitrary, we have
\[
  \forall z\compat v\, |A(t)|^{p(\eta(t))}_z .
\]
The premise induction hypothesis now yields \(|B|^r_y\), as required.

\medskip
\noindent
\textbf{Universal right.}
The case of universal right uses lifted Kleisli extension.  The rule is
\[
  \inferrule*[right=\(\forall\mathrm{R}\qquad x\notin \FV(\Gamma)\)]
  {\Gamma\vdash A(x)}
  {\Gamma\vdash \forall x^\sigma A(x)} .
\]
Let \(y: \sigma^\nucleus \times \ty^-_{A}\) be a challenge variable to the conclusion \(\forall x^\sigma A(x)\), and write \(y\equiv\pair{a}{b}\), where \(a:\sigma^\nucleus\) and \(b:\ty_A^-\).
By the induction hypothesis for the premise, with exact variable \(x:\sigma\) and challenge \(b\), we obtain a realizer \(r(x):\ty_A^+\)
and challenge bounds \(u_i(x):(\ty^-_{G_i})^\nucleus\) for the assumptions
\(G_i\in\Gamma\), such that
\[
  \forall \vec z\compat\vec u(x)\,|\Gamma|^{\vec g}_{\vec z}
  \vdash
  |A(x)|^{r(x)}_b .
\]
Here \(r(x)\) and the bounds \(u_i(x)\) may depend on \(x\).  Since
\(x\notin\FV(\Gamma)\), the realizer variables \(\vec g\) for the assumptions
in \(\Gamma\) do not depend on \(x\).
For the conclusion \(\forall x^\sigma A(x)\), define the realizer
\[
  h \eqdef \lke_{\sigma,A}(\lambda x^\sigma.\,r(x))
  : \sigma^\nucleus \to \ty_A^+ .
\]
For each assumption \(G_i\in\Gamma\), define the challenge bound
\[
  U_i\eqdef(\lambda x^\sigma.\,u_i(x))^\kappa(a)
  :(\ty^-_{G_i})^\nucleus .
\]
By Lemma~\ref{lem:uniform-bound}, applied componentwise to the induction
hypothesis for the premise, the bounds \(U_i\) are inhabited, assuming \(a\) is inhabited, and the translated assumptions for \(\Gamma\), bounded by \(\vec U\), yield
\[
  \forall x\compat a\, |A(x)|^{r(x)}_b .
\]
Now let \(x\compat a\).  From the displayed formula we have
\[
  |A(x)|^{r(x)}_b .
\]
Moreover, Lemma~\ref{lem:lke-propagation} gives
\[
  r(x)\nule_A h(a).
\]
By monotonicity in realizers, we obtain
\[
  |A(x)|^{h(a)}_b .
\]
Since \(x\compat a\) was arbitrary, this proves
\[
  |\forall x^\sigma A(x)|^h_{\pair{a}{b}} .
\]

\medskip
\noindent
\textbf{Existential left.}
The case of existential left uses lifted Kleisli extension to remove the witness parameter.  The rule is
\[
  \inferrule*[right=\(\exists\mathrm{L}\)]
  {\Gamma,A(x) \vdash B}
  {\Gamma,\exists x^\sigma A(x) \vdash B}
  \qquad x\notin\FV(\Gamma,B).
\]
Let \(p: \sigma^\nucleus \times \ty^+_{A}\) be a realizer variable for the
assumption \(\exists x^\sigma A(x)\), and write \(p\equiv\pair{p_1}{p_2}\),
where \(p_1:\sigma^\nucleus\) and \(p_2:\ty_A^+\).

By the induction hypothesis for the premise, with exact eigenvariable \(x:\sigma\), with \(p_2\) as the realizer variable for the assumption \(A(x)\), and with challenge \(y:\ty_B^-\),
we obtain a realizer \(r(x):\ty_B^+\), challenge bounds
\(u_i(x):(\ty^-_{G_i})^\nucleus\) for the assumptions \(G_i\in\Gamma\), and a
challenge bound \(v(x):(\ty_A^-)^\nucleus\) for the assumption \(A(x)\), such
that
\[
  \forall \vec z\compat\vec u(x)\,|\Gamma|^{\vec g}_{\vec z},
  \forall z\compat v(x)\,|A(x)|^{p_2}_z
  \vdash
  |B|^{r(x)}_y .
\]
The side condition \(x\notin\FV(\Gamma,B)\) ensures that the formulas in
\(\Gamma\) and the conclusion \(B\) do not contain \(x\), but the extracted
terms \(r(x)\), \(u_i(x)\), and \(v(x)\) may still depend on \(x\).

For the conclusion sequent, define the realizer for \(B\) by lifting the premise
realizers along the witness information:
\[
  r'\eqdef \lke_{\sigma,B}(\lambda x^\sigma.\,r(x),p_1):\ty_B^+ .
\]
For each assumption \(G_i\in\Gamma\), define the uniform challenge bound
\[
  U_i\eqdef (\lambda x^\sigma.\,u_i(x))^\kappa(p_1)
  :(\ty^-_{G_i})^\nucleus .
\]
For the assumption \(\exists x^\sigma A(x)\), first collect the challenge bounds
for the body over the possible witnesses:
\[
  V\eqdef (\lambda x^\sigma.\,v(x))^\kappa(p_1)
  :(\ty_A^-)^\nucleus .
\]
The challenge bound assigned to the existential assumption is then
\[
  W\eqdef \eta(V):((\ty_A^-)^\nucleus)^\nucleus .
\]

Assume the translated assumptions for \(\Gamma\), bounded by \(\vec U\), and the
translated existential assumption
\[
  H:\forall u\compat W\,|\exists x^\sigma A(x)|^p_u .
\]
By \textup{(IN1)}, \(V\compat\eta(V)=W\), so \(H\) gives
\[
  |\exists x^\sigma A(x)|^p_V .
\]
Unfolding the existential clause, we obtain
\[
  \exists x\compat p_1\,\forall z\compat V\,|A(x)|^{p_2}_z .
\]
Choose such an \(x\).  Since \(x\compat p_1\), \textup{(IN2)} applied to
\(\lambda x.\,v(x)\) gives
\[
  v(x)\nule V .
\]
Hence bounded monotonicity yields
\[
  \forall z\compat v(x)\,|A(x)|^{p_2}_z .
\]
Moreover, Lemma~\ref{lem:uniform-bound}, applied to the premise
induction hypothesis and the bounds \(\vec U\), gives the translated assumptions
for \(\Gamma\) bounded by \(\vec u(x)\).  Therefore the premise induction
hypothesis yields
\[
  |B|^{r(x)}_y .
\]
Finally, Lemma~\ref{lem:lke-propagation}, using \(x\compat p_1\), gives
\[
  r(x)\nule_B r' .
\]
By monotonicity in realizers, we conclude \(|B|^{r'}_y\) as required.
Inhabitedness of the bounds \(\vec U\) and \(W\) follows from the
induction hypothesis and the axioms (IN1) and (IN2).

\medskip
\noindent
\textbf{Existential right.}
The case of existential right uses exact information for the displayed
witness.  The rule is
\[
  \inferrule*[right=\(\exists\mathrm{R}\)]
  {\Gamma\vdash A(t)}
  {\Gamma\vdash \exists x^\sigma A(x)}.
\]
Let \(y: (\ty_A^-)^\nucleus\) be a challenge variable to the conclusion \(\exists x^\sigma A(x)\).
By the induction hypothesis for the premise, we have a realizer \(r:\ty_A^+\) for \(A(t)\) and challenge bounds \(\vec u\) for the assumptions in \(\Gamma\), such that
\[
  \forall \vec z\compat\vec u\,|\Gamma|^{\vec g}_{\vec z}
  \vdash
  |A(t)|^r_v
\]
where the bounds \(\vec u\) may depend on the challenge variable \(v:\ty_A^-\) to \(A(t)\).

For the conclusion sequent, we define the realizer for the conclusion \(\exists x^\sigma A(x)\) by
\[
  q\eqdef\pair{\eta(t)}{r}: \sigma^\nucleus \times \ty^+_A
\]
and define the challenge bounds \(U_i:(\ty^-_{G_i})^\nucleus\) for the assumptions \(G_i \in \Gamma\) by
\[
  U_i \equiv (\lambda v^{\ty_A^-}.\,u_i)^\kappa(y) : (\ty^-_{G_i})^\nucleus.
\]
By Lemma~\ref{lem:uniform-bound}, applied to the premise induction hypothesis, the bounds \(U_i\) are inhabited, assuming \(y\) is inhabited, and
\[
  \forall \vec z\compat \vec U\,|\Gamma|^{\vec g}_{\vec z}
  \vdash
  \forall v\compat y\,|A(t)|^r_v .
\]
By (IN1), \(t\compat\eta(t)\), so the conclusion of the above sequent becomes
\[
  \exists x \compat \eta(t) \, \forall v \compat y \, |A(t)|^r_v
\]
which is exactly the translation \(|\exists x^\sigma A(x)|^q_y\).

\paragraph*{Equality axioms and defining equations}
For the neutral source systems, the equality axioms and the defining equations
for terms carry no computational content.  Primitive formulas have trivial
realizer and challenge types, and their translation is the formula itself.  Thus
an atomic equational axiom, such as a defining equation for a term constant, is
realized by \(\star:\One\), and its translated conclusion is just the same
equation, available as an axiom of the target system.  The same applies to the
constructor separation and injectivity axioms at the base types.

The neutral equality axioms with logical structure, such as symmetry,
transitivity, and congruence for term contexts, are realized by the canonical
logical realizers.  For example, the congruence axiom
\[
  s =_\sigma t \to r[x\eqsubst s] =_\tau r[x\eqsubst t]
\]
is realized by the trivial implication realizer: it sends the unit realizer of
the antecedent to the unit realizer of the consequent, and assigns the canonical
challenge bound \(\cc_\One\) to the antecedent.  After unfolding the translation,
the required formula is exactly the corresponding congruence axiom in the target
system.

For the observational source system \(\OHAo\), only base-type equalities are
primitive.  Hence the base equality axioms and the base-type defining equations
are handled as above.  Equalities at product and function types are not primitive,
but abbreviations expanded by the clauses defining observational equality.
Therefore their treatment is by induction on the type of the equality:
product equality reduces to the two component equalities, and function equality
reduces to pointwise equality.  The realizing terms are again the canonical ones
provided by the logical clauses of the translation, while the target derivations
come from unfolding observational equality and using the already handled lower-type
cases.  No congruence or extensionality rule for arbitrary term contexts is used
for \(\OHAo\).

\paragraph*{Closed induction rule}
\textbf{Closed induction.}
Suppose the last rule is closed induction
\[
  \inferrule*[right=\(\mathrm{ind}\)]
  {\vdash A(0) \\
    \vdash \forall k^\N(A(k)\to A(\suc k))}
  {\vdash \forall n^\N A(n)}.
\]
By the induction hypotheses for the two premises, we have a base realizer \(r_0:\ty_A^+\) such that
\[
  \vdash |A(0)|^{r_0}_b
\]
for arbitrary challenge \(b:\ty_A^-\), and a step realizer
\( r_S : \N^\nucleus\to \ty^+_{A(k)\to A(\suc k)} \)
such that
\[
  \vdash |\forall k^\N(A(k)\to A(\suc k))|^{r_S}_d
\]
for all challenges \(d:\ty^-_{\forall k^\N(A(k)\to A(\suc k))}\).
For \(p:\N^\nucleus\), write \(r_S(p)\equiv \langle r_S^+(p),r_S^-(p)\rangle\),
where
\[
  r_S^+(p):\ty_A^+\to\ty_A^+
  \qquad\text{and}\qquad
  r_S^-(p):\ty_A^+\times\ty_A^-\to(\ty_A^-)^\nucleus .
\]
Unfolding the translation of the step premise, this is
\[
  \vdash \forall k \compat p \bigl( ( \forall z \compat r^-_S(p)(\pair{a}{u}) |A(k)|^a_z ) \to |A(\suc k)|^{r^+_S(p)(a)}_u \bigr)
\]
for any index bound \(p:\N^\nucleus\), realizer \(a:\ty_A^+\) for \(A(k)\), and challenge \(u:\ty_A^-\) to \(A(\suc k)\).

Define the realizer \(h:\N^\nucleus\to\ty_A^+\) for conclusion \(\forall n^\N A(n)\) by lifted recursion:
\[
  h(p) \eqdef
  \lrec_A \bigl(r_0,\, \lambda q^{\N^\nucleus}. r_S^+(q) ,\, p \bigr)
  : \ty_A^+ .
\]
Equivalently, for exact numbers \(n:\N\), define
\[
  H(n) \eqdef
  \rec_{\ty_A^+}
  \bigl(r_0, \, \lambda k^\N. r_S^+(\eta(k)) ,\, n\bigr) : \ty_A^+,
\]
so that \(h(p) \equiv \lke_{\N,A}(H,p)\).

We first prove, by closed induction in the target system, the auxiliary closed formula
\[
  \vdash \forall n^\N \, \forall v^{\ty_A^-} \, |A(n)|^{H(n)}_v .
\]
For \(n=0\), this is the conclusion supplied by the induction hypothesis for the base premise.
For the step, fix \(k:\N\) and assume
\[
  \forall v^{\ty_A^-} \, |A(k)|^{H(k)}_v
\]
and fix the challenge \(u:\ty_A^-\) to \(A(\suc k)\).
We need to prove \(|A(\suc k)|^{H(\suc k)}_u\).
We instantiate the unfolded step premise translation with the index bound \(\eta(k)\), the realizer \(H(k)\) for \(A(k)\), the challenge \(u\) to \(A(\suc k)\), and the index \(k\), and then obtain
\[
  \bigl(
  \forall z \compat r^-_S(\eta(k))(\pair{H(k)}{u})\, |A(k)|^{H(k)}_z
  \bigr)
  \to
  |A(\suc k)|^{r^+_S(\eta(k))(H(k))}_u
\]
because \(k \compat \eta(k)\) by (IN1).
The antecedent follows from the assumption \(\forall v^{\ty_A^-} \, |A(k)|^{H(k)}_v\),
and the consequent is precisely \(|A(\suc k)|^{H(\suc k)}_u\).
Thus closed induction proves the auxiliary formula.

It remains to verify the translated conclusion.
A challenge to \(\forall n^\N A(n)\) has the form \(\langle p,v\rangle\) with
an index bound \(p:\N^\nucleus\) and a challenge \(v:\ty^-_{A}\).  We must prove
\[
  \forall n\compat p\, |A(n)|^{h(p)}_v .
\]
Let \(n\compat p\).  From the auxiliary formula we have
\[
  |A(n)|^{H(n)}_v .
\]
By correctness of lifted recursion (Corollary~\ref{cor:lifted-rec}), we have
\[
  H(n) \nule_A h(p).
\]
By monotonicity in realizers (Lemma~\ref{lem:formula-monotonicity}), we obtain
\[
  |A(n)|^{h(p)}_v,
\]
as required.

\paragraph*{Contraction rule}
To obtain full soundness, it remains to verify the single non-affine structural rule used in the full soundness theorem (Theorem~\ref{thm:full-soundness}).

\medskip
\noindent
\textbf{Contraction.}
Suppose the last rule is contraction
\[
  \inferrule*[right=\(\mathrm{contr}\)]
  {\Gamma,A,A \vdash B}
  {\Gamma,A \vdash B}.
\]
By the induction hypothesis for the premise, we have a realizer \(r:\ty_B^+\), challenge bounds \(\vec w\) for the assumptions in \(\Gamma\), and two challenge bounds \(u,v:(\ty_A^-)^\nucleus\) for the two occurrences of \(A\), such that
\[
  \forall \vec z\compat \vec w \,|\Gamma|^{\vec x}_{\vec z},\,
  \forall z\compat u\,|A|^x_z,\,
  \forall z\compat v\,|A|^x_z
  \vdash
  |B|^r_b .
\]
For the contracted sequent \(\Gamma,A \vdash B\), the realizer for \(B\) and the bounds for
\(\Gamma\) are unchanged.  The single remaining occurrence of
\(A\) is assigned the challenge bound
\[
  \ctr_A(x,u,v) : (\ty_A^-)^\nucleus .
\]
It is inhabited by the inhabitedness preservation property of the contraction operation and the induction hypothesis for \(u\) and \(v\).
Now assume the single translated assumption
\[
  \forall z\compat \ctr_A(x,u,v)\,|A|^x_z .
\]
By the defining property of the contraction operation, it yields the conjunction
\[
  \bigl(\forall z\compat u\,|A|^x_z\bigr)
  \wedge
  \bigl(\forall z\compat v\,|A|^x_z\bigr).
\]
Using the two projections of this conjunction as the two translated assumptions
for the two occurrences of \(A\), the induction hypothesis for the premise gives
\(|B|^r_b\).

This derivation uses only ordinary logical reasoning in the target
system, namely implication elimination and conjunction elimination, together
with the defining property of \(\ctr_A\).  No structural contraction is applied in
the target: the two assumptions needed by the premise are produced as two
components of a conjunction from the single contracted translated assumption.

\section{An affine functional interpretation with continuity information}
\label{app:continuity-fi}

We finish by spelling out how the affine soundness theorem can make extracted
realizers carry continuity information.  We use the pointwise-continuity nucleus
from the parametrized translation of System~\(\T\)
\cite[Section~3.3]{Xu2020Gentzen}, adapted to the present realizer--challenge
setting.

Two mild adjustments to the abstract presentation are needed for this example.
First, the nucleus is oracle-parametrized: the information types and the
operations \(\eta\) and \(\kappa\) are defined uniformly, but compatibility is
interpreted relative to a fixed oracle parameter \(\alpha:\baire\).  Thus the
axioms \textup{(IN1)} and \textup{(IN2)} are proved uniformly in \(\alpha\).
Second, to read the continuity data carried by positive realizers, we use a
standard strengthened invariant of the soundness induction: extracted realizers
are inhabited, in the sense that the information objects occurring in their
positive components are compatible with suitable exact values.  This
inhabitedness invariant is not needed for the abstract correctness theorem, and
so we did not include it in the main statement.  In the present example,
however, it is exactly what shows that the modulus components inside the
extracted realizer are correct.

We now recall the pointwise-continuity nucleus in this oracle-parametrized form.
For each type \(\sigma\), define
\[
  \sigma^\nucleus \eqdef (\baire \to \sigma) \times (\baire \to \N).
\]
If \(a:\sigma^\nucleus\), write \(\val_a:\baire\to\sigma\) and \(\modu_a:\baire\to\N\) for its two components. The intended meaning is that \(\val_a(\alpha)\) is the value represented by \(a\) at the oracle \(\alpha : \baire\), while \(\modu_a(\alpha)\) is a local modulus witnessing that this value is stable under perturbations of \(\alpha\) below that modulus.
Accordingly, for each oracle \(\alpha:\baire\), define a compatibility relation \(\compat_\sigma^\alpha\) by
\[
  x \compat_\sigma^\alpha a
  \eqdef
  (x =_\sigma \val_a(\alpha)) \wedge
  \forall \beta (\alpha =_{\modu_a(\alpha)} \beta \to \val_a(\alpha) =_\sigma \val_a(\beta)),
\]
where \(\alpha =_m \beta\) means that \(\alpha\) and \(\beta\) agree on the first \(m\) values, i.e., \(\alpha(i) = \beta(i)\) for all \(i < m\).
Thus an exact object \(x:\sigma\) is compatible with \(a:\sigma^\nucleus\) at \(\alpha\) when \(a\) computes the value \(x\) at \(\alpha\), together with a pointwise modulus showing that this value is locally constant around \(\alpha\).

The nucleus operations are defined as follows:
\[
  \begin{aligned}
    \eta(x)
     & \eqdef
    \bigl\langle
    \lambda \alpha. x,\,
    \lambda \alpha. 0
    \bigr\rangle,
    \\
    f^\kappa(a)
     & \eqdef
    \bigl\langle
    \lambda \alpha. \val_{f(\val_a(\alpha))}(\alpha),\,
    \lambda \alpha. \max(\modu_{f(\val_a(\alpha))}(\alpha),\, \modu_a(\alpha))
    \bigr\rangle .
  \end{aligned}
\]
The unit \(\eta\) sends an exact object to the corresponding constant value with zero modulus. Then (IN1) is immediate from the definition.
For the Kleisli extension of \(f:\sigma\to\tau^\nucleus\) at input \(a:\sigma^\nucleus\), the value component first computes the exact input \(\val_a(\alpha)\), applies \(f\) to that input, and then reads the value component of the resulting information object at the same oracle \(\alpha\).
Its modulus component records the two sources of oracle dependence: the input information must be stable enough to keep \(\val_a(\alpha)\) fixed, and the output information \(f(\val_a(\alpha))\) must be stable at that fixed input.  Hence we take the maximum of \(\modu_a(\alpha)\) and \(\modu_{f(\val_a(\alpha))}(\alpha)\).
For \(\textup{(IN2)}\), suppose \(x\compat^\alpha_\sigma a\).  Then
\(x=\val_a(\alpha)\).  If \(\beta\) agrees with \(\alpha\) below
\(\modu_{f^\kappa(a)}(\alpha)\), then it agrees below \(\modu_a(\alpha)\), hence
\(\val_a(\beta)=\val_a(\alpha)\).  It also agrees below
\(\modu_{f(\val_a(\alpha))}(\alpha)\), so the output information
\(f(\val_a(\alpha))\) is stable at \(\alpha\).  Therefore \(f(x)\nule
f^\kappa(a)\) relative to \(\alpha\).

Thus the affine soundness construction can be run with this oracle-parametrized
continuity nucleus, uniformly in \(\alpha\).  The resulting interpretation has
the same realizer--challenge shape as in Theorem~\ref{thm:affine-soundness}, but
the information objects appearing in positive realizer components now carry both
value and modulus data.  Using the strengthened inhabitedness invariant
mentioned above, this data is certified: whenever such an information object is
shown to be inhabited relative to \(\alpha\), its second component is a valid
pointwise modulus at \(\alpha\).

To see this information explicitly, consider a proof of a formula of the form
\[
  \bigl(\forall n^\N\exists m^\N A(n,m)\bigr)\to \exists k^\N B(k),
\]
where \(A\) and \(B\) are primitive formulas.  A realizer for the antecedent has,
up to the trivial unit component, type
\[
  \N^\nucleus\to\N^\nucleus,
\]
while a realizer for the consequent has type \(\N^\nucleus\).  Hence the
realizer-transforming component extracted for the implication has type
\[
  R:(\N^\nucleus\to\N^\nucleus)\to\N^\nucleus .
\]
To recover the ordinary functional computed by this transformer, together with
its continuity modulus, we apply \(R\) to the generic sequence
\[
  \Omega:\N^\nucleus\to\N^\nucleus
\]
defined by
\[
  \Omega(a) \eqdef
  \bigl\langle
  \lambda \alpha.\,\alpha(\val_a(\alpha)),\,
  \lambda \alpha.\,\max(\modu_a(\alpha),\val_a(\alpha)+1)
  \bigr\rangle . 
\]
The value component of \(\Omega(a)\) queries the oracle at the index represented
by \(a\).  Its modulus says that this query is stable once the index is stable
and the oracle is fixed up to the queried position.

The resulting term \(R(\Omega):\N^\nucleus\) should be compared with what is extracted
from the same proof by the exact-information instance, that is, the usual Dialectica-style interpretation.
This comparison can be justified by a formula-indexed logical relation between realizers in these two instances; we only record the resulting correspondence needed here.
Write the realizer-transforming component of the exact-information realizer as
\[
  R_{\mathrm{ex}}:(\N\to\N)\to\N .
\]
For each oracle \(\alpha:\baire\) realizing the antecedent
\(\forall n^\N\exists m^\N A(n,m)\), in the sense that
\(A(n,\alpha(n))\) holds for every \(n\), the value component of the continuity
realizer agrees with the exact-information realizer:
\[
  \val_{R(\Omega)}(\alpha)=R_{\mathrm{ex}}(\alpha).
\]
The second component \(\modu_{R(\Omega)}:\baire\to\N\) carries the additional continuity information.
If \(\alpha\) realizes the antecedent, then \(\alpha\), viewed as an exact
antecedent realizer, is compatible with the generic sequence \(\Omega\) at
\(\alpha\): whenever an exact input \(n\) is compatible with input information
\(a:\N^\nucleus\), the exact witness \(\alpha(n)\) is compatible with the output
information \(\Omega(a)\).  Hence the translated antecedent is satisfied by
\(\Omega\).  Applying the inhabitedness-strengthened invariant of soundness to the extracted
positive realizer gives that the output information object \(R(\Omega)\) is
inhabited relative to \(\alpha\), that is, for some \(k:\N\),
\[
  k\compat^\alpha_\N R(\Omega).
\]
Unfolding the definition of \(\compat^\alpha_\N\) gives
\[
  k = \val_{R(\Omega)}(\alpha)
  \quad\text{and}\quad
  \forall \beta^{\baire}\,
  \bigl(
    \alpha =_{\modu_{R(\Omega)}(\alpha)} \beta
    \to
    \val_{R(\Omega)}(\beta)=\val_{R(\Omega)}(\alpha)
  \bigr).
\]
Therefore \(\modu_{R(\Omega)}(\alpha)\) is a pointwise modulus of continuity at
\(\alpha\) for the exact-information realizer-transformer, represented by the
functional \(\val_{R(\Omega)}\).

This example also explains why extracted realizers can carry additional
information beyond the ordinary exact-information realizers.  The extracted term
is built compositionally from the proof, using the realizer operations associated
with the logical and arithmetic rules of the interpretation.  Thus the syntactic
construction of the realizer can be used not only to compute ordinary witnesses,
but also to propagate additional data.  Our way of doing this is to choose an
information nucleus whose information types and operations already contain the
extra data, and whose compatibility relation states the corresponding
correctness property.  For the continuity nucleus, the information objects carry
both value and modulus components, and the operations of the interpretation
propagate these components together.  Consequently, the modulus component of an
extracted realizer is produced by the same proof extraction that produces the
ordinary realizer, rather than by a separate post-processing argument.

\end{document}